\documentclass[11pt]{article}

\usepackage{fullpage,amsfonts,color,amsmath,amsthm}

\usepackage{ifpdf}

\usepackage{cite}

\usepackage{amssymb}

\usepackage[normalem]{ulem} 
\usepackage[usenames,dvipsnames]{xcolor} 

\newcommand{\cl}[1]{\text{cl}\left(#1\right)}
\newcommand{\I}{\mathbf{I}}

\newcommand{\supp}{\textit{supp}}
\newcommand{\ess}{\text{ess}}

\newcommand{\edgw}{\dashrightarrow}

\newcommand{\prob}{\text{Pr}}
\newcommand{\Hh}{\mathbf{H}}
\newcommand{\E}[1]{\mathbf{E}\left[#1\right]}

\theoremstyle{plain}

\begin{document}

\def\a{\alpha}
\def\b{\beta}
\def\c{\chi}
\def\d{\delta}
\def\D{\Delta}
\def\e{\epsilon}
\def\f{\phi}
\def\F{\Phi}
\def\g{\gamma}
\def\G{\Gamma}
\def\k{\kappa}
\def\K{\Kappa}
\def\z{\zeta}
\def\th{\theta}
\def\Th{\Theta}
\def\l{\lambda}
\def\la{\lambda}
\def\La{\Lambda}
\def\m{\mu}
\def\n{\nu}
\def\p{\pi}
\def\P{\Pi}
\def\vp{\varphi}
\def\r{\rho}
\def\s{\sigma}
\def\S{\Sigma}
\def\t{\tau}
\def\om{\omega}
\def\Om{\Omega}
\def\U{\Upsilon}
\def\smallo{{\rm o}}
\def\bigo{{\rm O}}
\def\to{\rightarrow}
\def\E{{\bf Exp}}
\def\ex{{\bf Exp}}
\def\cd{{\cal D}}
\def\rme{{\rm e}}
\def\hf{{1\over2}}
\def\cala{{\cal A}}
\def\cale{{\cal E}}
\def\Fscr{{\cal F}}
\def\cc{{\cal C}}
\def\calc{{\cal C}}
\def\calh{{\cal H}}
\def\calf{{\cal F}}
\def\bk{\backslash}

\def\out{{\rm Out}}
\def\temp{{\rm Temp}}
\def\overused{{\rm Overused}}
\def\big{{\rm Big}}
\def\notbig{{\rm Notbig}}
\def\moderate{{\rm Moderate}}
\def\swappable{{\rm Swappable}}
\def\candidate{{\rm Candidate}}
\def\bad{{\rm Bad}}
\def\crit{{\rm Crit}}
\def\col{{\rm Col}}
\def\dist{{\rm dist}}
\def\poly{{\rm poly}}

\newcommand{\Exp}{\mbox{\bf Exp}}
\newcommand{\var}{\mbox{\bf Var}}
\newcommand{\pr}{\mbox{\bf Pr}}

\newtheorem{lemma}{Lemma}[section]
\newtheorem{theorem}[lemma]{Theorem}
\newtheorem{corollary}[lemma]{Corollary}
\newtheorem{claim}[lemma]{Claim}
\newtheorem{remark}[lemma]{Remark}
\newtheorem{observation}[lemma]{Observation}
\newtheorem{proposition}[lemma]{Proposition}
\newtheorem{definition}[lemma]{Definition}
\newtheorem{property}[lemma]{Property}
\newtheorem{example}[lemma]{Example}

\newcommand{\limninf}{\lim_{n \rightarrow \infty}}
\newcommand{\proofstart}{{\bf Proof\hspace{2em}}}
\newcommand{\tset}{\mbox{$\cal T$}}
\newcommand{\proofend}{\hspace*{\fill}\mbox{$\Box$}\vspace{2ex}}
\newcommand{\bfm}[1]{\mbox{\boldmath $#1$}}
\newcommand{\reals}{\mbox{\bfm{R}}}
\newcommand{\expect}{\mbox{\bf Exp}}
\newcommand{\he}{\hat{\e}}
\newcommand{\card}[1]{\mbox{$|#1|$}}
\newcommand{\rup}[1]{\mbox{$\lceil{ #1}\rceil$}}
\newcommand{\rdn}[1]{\mbox{$\lfloor{ #1}\rfloor$}}
\newcommand{\ov}[1]{\mbox{$\overline{ #1}$}}
\newcommand{\inv}[1]{\frac{1}{ #1}}

\def\calc{{\cal C}}
\def\cald{{\cal D}}
\def\calp{{\cal P}}

\title{Frozen variables in random boolean constraint satisfaction problems}

\author{Michael Molloy and Ricardo Restrepo\\
Department of Computer Science, University of Toronto\\
10 King's College Road, Toronto, ON}

\maketitle

\begin{abstract} We determine the exact {\em freezing threshold}, $r^f$, for a family of models of random boolean constraint satisfaction problems, including NAE-SAT and hypergraph 2-colouring, when the constraint size is sufficiently large. If the constraint-density of a random CSP, $F$, in our family is greater than $r^f$ then for almost every solution of $F$, a linear number of variables are {\em frozen}, meaning that their colours cannot be changed by a sequence of alterations in which we change $o(n)$ variables at a time, always switching to another solution.  If the constraint-density is less than $r^f$, then almost every solution has $o(n)$ frozen variables.

Freezing is a key part of the clustering phenomenon that is hypothesized by non-rigorous techniques from statistical physics. The understanding of clustering has led to the development of advanced heuristics such as Survey Propogation.  It has been suggested that the freezing threshold is a precise algorithmic barrier: { that} for densities below $r^f$ the random CSPs can be solved using very simple algorithms, while for densities above $r^f$ one requires more sophisticated techniques in order to deal with frozen clusters.
\end{abstract}

\newpage

\section{Introduction}

The clustering phemonenon is arguably the most important development in the study of random constraint satisfaction problems (CSP's) over the past decade or so.  Statistical physicists have discovered that for typical models of random constraint satisfaction problems, the structure of the solution space appears to undergo remarkable changes as the constraint density increases.  

{A common geometric interpretation of the clustering analysis paints the following picture.  Most of it is not proven rigorously; in fact many details are not specified precisely.  Nevertheless, there is evidence that something close to this takes place for many natural random CSP's:}
At first,
all solutions are very similar in that we can change any one solution into any other solution via
a sequence of small local changes; i.e. by changing only $o(n)$ variables-at-a-time, always having a satisfying solution.  This remains true for almost all solutions until the {\em clustering threshold}\cite{mpz,mz}, at which point they shatter into an exponential number of clusters.  Roughly speaking: one can move from any solution to any other {\em in the same cluster} making small local changes, but moving from one cluster to another requires changing a linear number of variables in at least one step.  As we increase the density further, we reach the {\em freezing threshold}\cite{gs}.  Above that point, almost all clusters\footnote{By this we mean: all but a vanishing proportion of the clusters, when weighted by their size.} contain {\em frozen variables}; that is, variables whose values do not change for any solutions in the cluster.  At higher densities, we find other thresholds, such as the {\em condensation threshold}\cite{kmrsz}  above which the largest cluster contains a positive proportion of the solutions.  Eventually we reach the {\em satisfiability threshold}, the point at which there are no solutions.  

The methods that are used to describe these phenomena and determine the values of the thresholds are mathematically sophisticated, but are typically not rigorous.  Nevertheless, they have transformed the rigorous study of random CSP's.  

For one thing, this picture explained things that mathematicians had already discovered.  For some problems (eg. $k$-NAE-SAT\cite{amnae}, $k$-SAT\cite{ap} and $k$-COL\cite{an}) the second moment method had been used to prove the existence of solutions at densities that are close to, but not quite, the hypothesized satisfiability threshold. We now understand that this is because the way that the second moment method was applied cannot work past the condensation threshold.  As another example, it had long been observed that at a point where the density is still far below the satisfiability threshold, no algorithms are proven to  find solutions for many of the standard random CSP models.   We now understand\cite{mz} that this observed ``algorithmic barrier'' is asymptotically equal to the  clustering threshold as $k$ grows{ (\cite{ac} provides rigorous grounding for this),} and so the difficulties {appear to arise from} the onset of clusters.   It has been suggested that this algorithmic barrier may occur precisely at the freezing threshold; i.e. the formation of clusters does not  cause substantial algorithmic difficulties until {most of} the clusters have frozen variables
(see section~\ref{sab} below).

Although the { picture described above} is, for the most part, not established rigorously,  understanding it has led to substantial new theorems\cite{aco,mmw,sly,gmont, mrt, cop, cobp, coe, ffv, abbe, ksv, amm1,amm2}.  For example, \cite{cop} used our understanding of how condensation has foiled previous second moment arguments to modify those arguments
and obtain a remarkably tight bound on the satisfiability threshold for $k$-NAE-SAT.  \cite{cobp} used our understanding of clustering to design an algorithm that provably solves random $k$-SAT
up to densities of $O(2^k\ln k/k)$, which is the asymptotic value of the clustering threshold.
A particularly impressive heuristic result is the Survey Propogation algorithm\cite{mz,bmz}, which experimentally has solved random 3-SAT on $10^7$ variables at densities far closer to the satisfiability threshold than anyone
had previously been able to handle, even on fewer than 1000 variables. This algorithm was designed specifically to take advantage of the clustering picture.

Of course, another thrust has been to try to rigorously establish pieces of the clustering picture\cite{ac,acr,mrt,xor1,xor2,cz,zkqp,mmcol}.
  We have been most successful with $k$-XOR-SAT; i.e. a random system of boolean linear equations.  The satisfiability threshold was established in \cite{dm} for $k=3$ and in \cite{cuc} for $k\geq 4$.  More recently, \cite{xor1,xor2} each established a very precise description of the clustering picture.  It should be noted that the solutions of a system of linear equations are very well-understood, and that was of tremendous help in the study of the clustering of the solutions.
Other CSP's, for which we do not have nearly as much control over the solutions, have been {much} more resistant to rigorous analysis; {nevertheless, there have been substantial results - see Section \ref{srw}.}

The contribution of this paper is to rigorously determine the precise freezing threshold for a  family
of CSP models including $k$-NAE-SAT and hypergraph 2-colouring.  The freezing threshold for
$k$-COL was determined by the first author in \cite{mmcol}; prior to this work, $k$-COL and $k$-XOR-SAT are the only two common models for which the freezing threshold was determined rigorously.  

We follow the approach of \cite{mmcol},
but we differ mainly in: (i) Where\cite{mmcol} analyzed the {\em Kempe-core}, we need to analyze the {\em *-core}, which was introduced in \cite{art} to prove the existence of frozen variables in random $k$-SAT. (ii) Rather than carrying out the analysis for a single model,
we carry it out simultaneously for a family of models. 

Our informal description of freezing described it in terms of the clusters.  At this point,  not enough information about clustering has been established rigorously to permit us to define
freezing in those terms.  (Eg. we do not know the {exact} clustering threshold for any interesting model except $k$-XOR-SAT.) So our formal definition of a frozen variable avoids the notion of clustering.

\begin{definition} An $\ell$-path of solutions of a CSP $F$ is a sequence
$\s_0,\s_1,...,\s_t$ of solutions, where for each $0\leq i\leq t-1$, $\s_{i}$ and $\s_{i+1}$ differ on at most $\ell$ variables.
\end{definition}

\begin{definition}\label{df}
 Given a solution $\s$ of a CSP $F$, we say that a variable $x$ is {\em $\ell$-frozen} with respect to $\s$ if for every $\ell$-path
$\s=\s_0,\s_1,...,\s_t$ of solutions of $F$, we have $\s_t(x)=\s(x)$.
\end{definition}

In other words, it is not possible to change the value of $v$ by changing at most $\ell$ vertices at a time.  Roughly speaking, the solutions in the same cluster as $\s$ are the solutions that can be reached by a $o(n)$-path.  So $x$ is $o(n)$-frozen with respect to $\s$ if $x$ has the same value in every solution in the same cluster as $\s$. Thus, this definition is essentially equivalent to the informal one if the clustering picture is accurate.

We make critical use of the {\em planted model} (section~\ref{sp}); {\cite{ac} permits us to do so}.
We prove that one can use the planted model up to a certain density, and so we want the freezing threshold to be below that density. It will be if the constraint size $k$ is sufficiently large; $k\geq 30$ will do.

We analyze CSP-models satisfying certain properties: non-trivial, feasible, symmetric,  balance-dominated, and 1-essential (defined in section~\ref{scsp}). The first four are needed to permit
the planted model; the fifth allows us to focus on the *-core.
Given such a CSP model $\U$, we define  constants $r_f(\U), r_p(\U)$ and function $\la(\U,r)$ below.  Our main theorem is that $r_f(\U)$ is the freezing threshold for $\U$ and that $\la(\U,r)$ is the proportion of frozen vertices. We require the density to be below $r_p(\U)$ in order to apply the planted model.  This is not just a technicality - if the density is  significantly above $r_p(\U)$, then it will be above the condensation threshold
and the {expressions} that we provide will fail to {yield the correct constants.}

Given a CSP-model $\U$, $C(\U,n,M)$ is a random instance of $\U$ on $n$ variables and with $M$ constraints (see Section~\ref{scsp}). We say that a property holds w.h.p. (with high probability)
if it holds with probability tending to 1 as $n\rightarrow\infty$.


\begin{theorem}\label{mt} Consider any non-trivial, feasible, symmetric,  balance-dominated, and 1-essential CSP-model $\U$ with $r_f(\U)<r_p(\U)$.
Let $\s$ be a uniformly random solution of $C(\U,n,M=r n)$.
 \begin{enumerate}
  \item[(a)] For any $r_f(\U)<r<r_p(\U)$, there exists a constant $0<\b<1$ for which:
\begin{enumerate}
\item[(i)] w.h.p. there are $\la(\U,r)n +o(n)$ variables that are $\b n$-frozen with respect to $\s$.
\item[(ii)] w.h.p. there are $(1-\la(\U,r))n +o(n)$ variables that are not $1$-frozen with respect to $\s$.
\end{enumerate}
\item[(b)] For any $r<r_f(\U)$, 
w.h.p. at most $o(n)$ variables  are {$1$-frozen} with respect to $\s$.
\end{enumerate}
\end{theorem}

In other words, {in a typical solution}: for $r>r_f$, a linear number of variables are $\a n$-frozen, while for $r<r_f$, all but at most $o(n)$ variables are not even $1$-frozen. Furthermore, for $r>r_f$ we specify the  number of $\a n$-frozen vertices,
up to an additive $o(n)$ term.  All but at most $o(n)$ of the other vertices are not even $1$-frozen.  

We remark that for $k$-COL and $k$-XOR-SAT, we have ``$\omega(n)$-frozen'' rather
than ``1-frozen'', for some $\omega(n)\rightarrow\infty$.  
{Part} (b) probably remains true upon replacing ``$o(n)$'' with ``zero''. 
The $o(n)$ terms arises from a limitation of using the planted model.

For $k\geq30$ we always have $r_f(\U)<r_p(\U)$ {(see Proposition~\ref{p27})} and so our theorem applies.

For densities below the freezing threshold, our proof yields that, in fact, almost all variables can be changed via a $o(n)$-path of length 1:

\begin{theorem}\label{mt2} Consider any non-trivial, feasible, symmetric,  balance-dominated, and 1-essential CSP-model $\U$ with with $r_f(\U)<r_p(\U)$
Let $\s$ be a uniformly random solution of $C(\U,n,M=r n)$ with $r<r_f(\U)$. 

 For any $\omega(n)\rightarrow\infty$,
w.h.p. for all but at most $o(n)$ variables $x$, there is a solution $\s'$ such that
(i) $\s'(x)\neq s(x)$ and (ii) $\s'(x),\s(x)$ differ on at most $\omega(n)$ variables.
\end{theorem}

As mentioned above, our theorems apply to $k$-NAE-SAT and hypergraph 2-colouring, two of the standard benchmark models.   $k$-NAE-SAT is a $k$-CNF boolean formula which is satisfied if every clause contains at least one true literal and at least one false literal.  For hypergraph 2-colouring, we are presented with a $k$-uniform hypergraph and we need
to find a boolean assignment to the vertices so that no hyperedge contains only vertices of one sign.  Thus, it is equivalent to an instance of $k$-NAE-SAT where every literal is signed positively.  See Appendix~\ref{a1} for a discussion of other models to which our theorems apply.

{Physicists tell us that there is a second freezing threshold, above which {\em every} solution
has frozen variables\cite{gs,zk} (as opposed to {\em almost every} solution as in Theorem~\ref{mt}).
\cite{art} proves that this occurs in $k$-SAT for large enough densities (albeit for a weaker notion of freezing); see Section~\ref{srw}.   We do not see how to determine the exact value of that threshold.}

We should emphasize that the clustering picture described above is very rough.  The mathematical analysis used by statistical physicists to determine the various thresholds actually studies properties of certain Gibbs distributions on infinite trees rather than solutions of random CSP's.  The clustering picture is a common geometric interpretation and  it is not exact.
Nevertheless, there is very strong evidence that something {close} to this picture should hold.

\subsection{The algorithmic barrier}~\label{sab} A great deal of the interest in random CSP's arises from the long-established observation that as the densities approach the satisfiability threshold, the problems appear to be extremely difficult to solve\cite{ckt,msl}.  Much work has gone into trying to understand what exactly causes dense problems to be so algorithmically challenging (eg. \cite{csres,abmol,cobp,mz,ms}).

 It has been suggested (eg. \cite{z,zk,kk,kk2,drz,gs}) that, for typical CSP's,
the freezing threshold forms an algorithmic barrier.  For $r<r_f$ very simple algorithms (eg. greedy algorithms with minor backtracking steps) will w.h.p. find a satisfying solution, but for $r>r_f$ one requires much more sophisticated algorithms (eg. Survey Propogation). It has been proposed that the following simple algorithm should succeed for $r<r^f$:

Suppose that Theorem~\ref{mt2} were to hold for {\em every} solution $\s$.  We build our CSP one random constraint at a time, letting $F_i$ denote the CSP with $i$ constraints.  We begin with a solution $\s_0$ for $F_0$ ($\s_0$ can be any assignment). Then we obtain $\s_{i+1}$ from $\s_i$ as follows:  If $\s_i$ does not violate the $(i+1)$st constaint added, then we keep
$\s_{i+1}=\s_i$. Otherwise, we modify $\s_i$ into another solution $\s'$ of $F_i$
in which the values of the variables in the $(i+1)$st constraint are changed so that it is satisfied; then we set $\s_{i+1}=\s'$.   If Theorem~\ref{mt2} holds for $\s_i$,  then we can change each
of the $k$ variables in that constraint by changing only $o(n)$ other variables.  Expansion properties of a random CSP imply that these small changes will (usually) not interfere with each
other and so we can change each of the $k$ variables to whatever we want.  Thus we will eventually end up with a solution $\s_M$ to our random CSP $F_M$.

However, Theorem~\ref{mt2} does not hold for {\em every} solution, only most of them.  This is not just a limit of our proof techniques - it is believed that it does not hold for an exponentially small, but positive, proportion of the solutions.  So proving that this algorithm works would require showing that we never encounter one of those solutions.

To see, intuitively, why the onset of freezing may create algorithmic difficulties, consider {\em near-solutions} - assignments which violate only a small number of constraints, say $o(n)$ of them.  The near-solutions will also form clusters (because of {\em high energy barriers}; see \cite{ac}).  Furthermore, almost all clusters of near-solutions will not contain any solutions. This is because, above the freezing threshold, almost all solution clusters have a linear number of frozen variables and so after adding only $o(n)$ constraints, we will pick a constraint that violates the frozen variables. This will violate all solutions in that cluster, thus forming a near-solution cluster that contains no actual solutions. Of course, this description is non-rigorous but it provides a good intuition.

Now consider a greedy algorithm with  backtracking.  As it sets its variables, it will approach a near-solution $\r$.  At that point, it cannot move to a near-solution in a different cluster than $\r$, without employing a backtracking step that changes a linear number of variables.  So the algorithm will need to be sophisticated enough to approach one of the rare near-solution clusters that contains solutions.  

{As described above,} there is a second freezing threshold, above which {\em every} cluster has frozen variables. \cite{z} suggests that this is another algorithmic barrier above which even the sophisticated
algorithms fail to find solutions.  One indication is that, empirically, every solution $\s$ found by Survey Propogation is such that no variables are frozen with respect to $\s$.  So somehow, the algorithm is drawn to those rare unfrozen clusters, and hence may fail when there are no such clusters.

\subsection{Related work}\label{srw} The clustering picture for $k$-NAE-SAT and hypergraph 2-colouring was analyzed non-rigorously in \cite{drz}.  There are hundreds of other papers from the statistical physics community analyzing clustering and related matters.  Some are listed above; rather than listing more, we refer the reader to the book\cite{mmbook}. 

Achlioptas and Ricci-Tersenghi\cite{art} were the first to rigorously prove that freezing occurs  in a random CSP.  They studied random $k$-SAT and showed that for $k\geq 8$, for a wide range of edge-densities
below the satisfiability threshold and for {\em every} satisfying assignment $\s$, the vast majority of variables are 1-frozen w.r.t $\s$.  
They did so  by stripping down to the *-core, which inspired us to do the same here.
One difference between their approach and ours is that the variables of the *-core are 1-frozen by definition, whereas much of the work in this paper is devoted to proving that, for our models, they are in fact $\Theta(n)$-frozen.  We expect that our techniques should be able to prove
that the 1-frozen variables established in \cite{art} are, indeed, $\Theta(n)$-frozen.

\cite{ac}  proves the asymptotic (in $k$) {density} for the appearance of what they call {\em rigid} variables in  {$k$-COL, $k$-NAE-SAT and hypergraph 2-colouring (and proves that this is an upper bound for $k$-SAT)}.   The definition of rigid is somewhat weaker than frozen, but a simple modification extends their proof to show the same  for frozen vertices. So \cite{ac} provided the asymptotic, in $k$, location of the freezing threhold for {those models}.
\cite{mmcol} provided the exact location of the threshold for $k$-COL, when $k$ is sufficiently large.

\cite{acr,ac,mrt} establish the existence of what they call {\em cluster-regions} for {$k$-SAT, $k$-COL, $k$-NAE-SAT and hypergraph 2-colouring}. \cite{ac} proves that by the time the density exceeds $(1+o_k(1))$ times the hypothesized
clustering threshold the solution space w.h.p.\ shatters into an exponential number of $\Theta(n)$-separated
cluster-regions, each containing an exponential number of solutions. While these cluster-regions are not shown to be well-connected,
the well-connected property  does not seem to be crucial to the difficulties
that clusters pose for algorithms.  So {this} was a very big step towards explaining why an algorithmic barrier seems to arise asymptotically {(in $k$)} close to the clustering threshold.

\cite{amnae,am2col} provided the first asymptotically tight lower bounds on the satisfiability threshold
of $k$-NAE-SAT and hypergraph 2-colouring, achieving a bound that is roughly equal to the condensation threshold.
\cite{cz} provides an even stronger bound for hypergraph 2-colouring, extending above the condensation threshold.  \cite{cop} provides a remarkably strong bound for $k$-NAE-SAT - the difference between their upper and lower bounds decreases exponentially with $k$.

\section{CSP models}\label{scsp}
A  \emph{boolean constraint} of arity $k$ consists of $k$ \emph{ordered} variables $(x_1,\ldots,x_k)$ together with a boolean function $\varphi:\{-1,1\}^k\to \{0,1\}$. This function constrains the set of variables to take values $\sigma=(\sigma_1,\ldots,\sigma_k)\in\{-1,1\}^k$ such that $\varphi(\sigma_1,\ldots,\sigma_k)=1$.  We say that the constraint is {\em satisfied} by a boolean assignment $\s$ if it evaluates to 1 on $\s$.

A \emph{constraint satisfaction problem (CSP)} is a set  of  constraints, where the $a^{\text{th}}$ constraint is formed by a boolean function $\varphi_a$ over the variables $(x_{i_{1,a}},\ldots,x_{i_{k,a}})$, with $i_{j,a} \in [n]$. A CSP, $H$, defines a boolean function $F^{(H)}:\{-1,1\}^n \to \{0,1\}$ given by 
\[
F^{(H)}(\sigma_1,\ldots,\sigma_n):=\prod_{a} \varphi_a(\sigma_{i_{1,a}},\ldots,\sigma_{i_{k,a}}) .
\]
Given $\sigma\in\{-1,1\}^n$, we say that $\sigma$ is a \emph{satisfying assignment}, or {\em solution}, of the CSP $H$ if $\s$ satisfies every constraint of $H$; i.e. if $F^{(H)}(\sigma)=1$.  

A  {\em CSP model} is a set $\Phi$ of boolean functions, together with a probability distribution $p:\Phi\to [0,1]$ defined on it (we assume implicitly that the support of $p$ is $\Phi$).  
Our random CSPs are:
\begin{definition}
Given a CSP model $\U=(\Phi,p)$, a \emph{random CSP}, $C(\U,n,M)$, is a CSP over the variables $\{x_1,\ldots,x_n\}$ consisting of $M$ constraints $\{\varphi_{{a}}(x_{i_{1,a}},\ldots,x_{i_{k,a}}):a=1,\ldots,M\}$ where the boolean constraints $\{\varphi_{{a}}:a=1,\ldots,M\}$ are drawn independently from $\Phi$ according to the distribution $p$, and the $k$-tuples $\{(x_{i_{1,a}},...,x_{i_{k,a}}): a=1,\ldots,m\}$ are drawn uniformly and independently from the set of $k$-tuples of $\{x_1,\ldots,x_n\}$. 
\end{definition}


We consider random CSP-models $\U=(\Phi,p)$ with the following properties.

\begin{definition}\label{def:boolprop}\mbox{ }\\
\textbf{Non-trivial}: There is at least one $\varphi\in\Phi$ that is not satisfied by
$x_1=...=x_k=1$ and  at least one $\varphi\in\Phi$ that is not satisfied by
$x_1=...=x_k=-1$.\\
\\
\textbf{Feasible}: For any $\varphi\in\Phi$, and every assignment to any $k-1$ of the variables, at least one of the two possible assignments to the remaining variable will result in $\vp$ being satisfied.\\
\\
\textbf{Symmetric}: For every $\varphi\in\Phi$, and for every assignment $x$, we have $\varphi(x)=\varphi(-x)$, where  $-x$ is the assignment obtained from $x$ by reversing the assignment to each variable.\\
\\ 
\textbf{Balance-dominated} Consider a random assignment $\s$ where each variable
is independently set to be 1 with probability $q$ and -1 with probability $1-q$, and let $\vp$ be a random constraint from $\Phi$ with distribution $p$. The probability that $\s$ satisfies $\vp$ is maximized at $q=\hf$.
\end{definition}

Those four properties will allow us to apply the planted model.
`Non-trivial' is a standard property to require.  `Feasible' is also quite natural, although
some {common} models do not satisfy it.  The other two properties help us to bound the second moment
of the number of solutions, which in turn enables us to use the planted model.

Our final property allows us to analyze frozen variables using the *-core.

\begin{definition}{\bf 1-essential}:
Given a boolean constraint $\varphi$ and an assignment $\sigma$  that satisfies
$\varphi$, we say that the variable $x$ is {\em  essential} for $(\varphi,\sigma)$ if changing the value of $x$ results in $\varphi$ being unsatisfied.
We say that a set $\Phi$ of constraints is {\em 1-essential} if for every  $\varphi\in\Phi$, and every $\s$ satisfying $\varphi$,
at most one variable is essential for $(\varphi,\s)$.  A CSP-model $(\Phi,p)$ is  {\em 1-essential} if $\Phi$ is   1-essential. {A CSP is 1-essential if all of its constraints are 1-essential.}
\end{definition}

For example: in hypergraph 2-colouring, $x$ is essential iff its value is different from that of every other variable in $\phi$; in $k$-XOR-SAT, every variable is essential. It is easily confirmed that  {for $k\geq 3$}: $k$-SAT, hypergraph 2-colouring and $k$-NAE-SAT are  1-essential, but $k$-XOR-SAT is not.





\section{The planted model}\label{sp}
Consider any CSP-model $\U=(\Phi,p)$.  Theorem~\ref{mt} concerns a uniformly random satisfying assignment of
$C(\U,n,M)$; i.e. a pair $(F,\s)$ drawn from:
\begin{definition}
The {\em uniform model} $U(\U,n,M)$ is a random pair $(F,\s)$ where $F$ is taken from the $C(\U,n,M)$ model and $\s$ is a uniformly random satisfying solution of $F$.
\end{definition}

The uniform model is very difficult to analyze directly.  So instead we turn to the much more amenable planted model:

\begin{definition}\label{dplant} The  {\em planted model} $P(\U,n,M)$ is a random pair $(F,\s)$ chosen as follows:  Take a uniformly random assignment $\s\in\{-1,+1\}^n$. Next select a random $F$ drawn from $C(\U,n,M)$ conditional on $\s$ satisfying $F$. 
\end{definition}

{\bf Remark:} Note that we can select $F$ by choosing $M$ independent constraints.  Each time, we choose a uniformly random $k$-tuple of $k$ variables, then choose for those variables a constraint $\vp\in\Phi$ with probability distribution $p$. If $\s$ does not satisfy the constraint then reject and choose a new one. Equivalently, we can choose the $k$-tuples non-uniformly where the probability that
a particular $k$-tuple is chosen is proportional to the probability that, upon choosing $\vp$ for that
set, the constraint will be satisfied by $\s$. Then we choose $\vp\in\Phi$ with probability $p$ conditional on $\s$ satisfying $\vp$.

It is not hard to see that the uniform and planted models are not equivalent.  In the planted model, a CSP  is selected with probability roughly proportional to the number of satisfying assignments. Nevertheless,  Achlioptas and Coja-Oghlan\cite{ac}  proved that, under certain conditions, one can transfer results about the planted model to the uniform model when $\U$ is $k$-COL, $k$-NAE-SAT or hypergraph 2-colouring (also $k$-SAT, but under stronger conditions).  Montanari, Restrepo and Tetali\cite{mrt} extended this to all $\U$ in a class of CSP-models, including all models that are non-trivial, feasible, symmetric, and balance-dominated.

For each non-trivial, feasible, symmetric and balance-dominated CSP-model $\U$ we define (in Appendix~\ref{a1})
a constant $r_p(\U)$, which is the highest density for which we can use the planted model.  
The following key tool essentially follows from Theorem B.3 of \cite{mrt}, {except that they do not explicitly mention $r_p(\U)$, instead giving an implicit lower bound under appropriate conditions.} It was first proven in \cite{ac} for NAE-SAT, hypergraph 2-COL and a few other models.  
 
\begin{lemma} \label{tptp}  Consider any non-trivial, feasible, symmetric, and balance-dominated CSP-model $\U$.  
For  every $r<r_p(\U)$, there is
a function $g(n)=o(n)$ such that: Let $\cale$ be any property of pairs $(F,\s)$ where $\s$ is a satisfying solution of $F$. 
If
\[\pr(P(\U,n,M=rn) \mbox{ has } \cale)>1-e^{-g(n)},\]
then
\[\pr(U(\U,n,M=rn) \mbox{ has } \cale)>1-o(1).\]
\end{lemma}

In Appendix~\ref{a1}, we prove that if $\U$ is also 1-essential, then for $k\geq30$, we have $r_p(\U)>r_f(\U)$ and so Theorem~\ref{mt} is non-trivial. In fact, $r_p(\U)=\Theta(\frac{k}{\ln k})r_f(\U)$.  
The bound $k\geq 30$ can be lowered, and for some specific models $\U$ it can
be lowered significantly.  For example, for $k$-NAE-SAT and hypergraph 2-colouring,
one can probably prove that $k\geq 6$ will do.

\section{The *-core}\label{sstar}  The *-core was introduced in \cite{art} to study frozen variables in random $k$-SAT.

Fix a satisfying assignment $\s$, and consider a variable $x$.  Suppose that there are no constraints $\vp$ such that $x$ is essential for $(\vp,\s)$.  Then, by the definition of essential, we can change $x$ and still have a satisfying assignment.  So $x$ is not frozen.  This inspires the following:

\begin{definition}  Consider a CSP $F$ with a satisfying assignment $\s$. The {\em *-core} of $(F,\s)$ is the sub-CSP formed as follows:\\
Iteratively remove every variable $x$ such that for every constraint {$\vp$}:  $x$ is not essential for $(\vp,\s)$. When we remove a variable, we also remove all constraints containing that variable.  
\end{definition}
Note that the order in which variables are deleted will not affect the outcome of the iterative procedure. So the *-core is well-defined, albeit possibly empty. 

As described above, it is clear that the first variable removed is not frozen.  Expansion properties of a random CSP - in particular the fact that it is locally tree-like - imply that almost every variable removed is not frozen. Furthermore, we will prove that if the model is 1-essential then almost all variables that remain in the *-core are frozen.  Having proven those two key results, Theorem~\ref{mt} follows from an analysis of the *-core process.

Now suppose that our CSP-model is 1-essential.
A key observation is that the *-core depends only on the constraints that have essential variables.  I.e., if we first remove all constraints with no essential variables from the CSP and then apply the *-core process, the set of {variables} in the resultant *-core will not change.  

\begin{definition}
Given a 1-essential CSP, $F$, and a satisfying solution $\s$, we define {\em the hypergraph
$\G(F,\s)$} as follows:  The vertices are  the variables of $F$ and the  variables of each  constraint of $F$ form a hyperedge, if that constraint has an essential variable.  That essential variable
is called the {\em essential vertex} of the  hyperedge.
\end{definition}

Note that we can find the *-core of $(F,\s)$ by repeatedly deleting from $\G(F,\s)$
 vertices that are not essential in any hyperedges, {along with all hyperedges containing the deleted vertices}. The resulting hypergraph is called the {\em *-core of $\G(F,\s)$}.

The precise model for the  random hypergraph $\G(F,\s)$ varies with $\U$ (see appendix~\ref{asc}).  However,
 the size of the *-core as a function of the number of hyperedges is the same
for all such models.

We define:
\[\a_k:=\inf_{x> 0}\frac{x}{(1-e^{-x})^{k-1}}.\]

Also, for $\alpha>\alpha_k$, let $x_k(\a)$ be the maximum value of $x\geq 0$ such that 
$\frac{x}{(1-e^{-x})^{k-1}}=\alpha$ and set
\[\r_k(\a)=1-e^{-x_k(\a)}.
\]
In Appendix~\ref{a3}, we prove

\begin{lemma}\label{lHp} Consider any 1-essential CSP-model $\U=(\Phi,p)$ of arity $k$,
and a random CSP, $F$, drawn from $P(\U,n,M=rn)$.  Suppose $\G(F,\s)$ has $\a n +o(n)$ hyperedges. For any $g(n)=o(n)$,  with probability at least $1-e^{-g(n)}$:
\begin{enumerate}
\item[(a)] If $\a> \a_k$ then the *-core of  $\G(F,\s)$ has $\r_k(\a) n+o(n)$ vertices.
\item[(b)] If $\a< \a_k$ then the *-core of  $\G(F,\s)$ has $o(n)$ vertices.
\end{enumerate}
\end{lemma}

This allows us to analyze our family of models simultaneously by working directly with the *-core of $\G(F,\s)$.  We prove that almost all vertices of the *-core are $\Theta(n)$-frozen variables in $F$
and almost all vertices outside of the *-core are not even 1-frozen in $F$.

 In Appendix~\ref{a2}, we define for any 1-essential  CSP-model
$\U=(\Phi,p)$, a constant $\xi(\U)>0$ and prove:
\begin{lemma}\label{lxi}
For any $g(n)=o(n)$ and $r>0$, with probability at least $1-e^{-g(n)}$, the number of constraints
in $P(\U,n,M=rn)$ that have an essential variable is $\xi(\U) r n +o(n)$.
\end{lemma}
This yields Theorem~\ref{mt} (see appendix ~\ref{asc}) with:
\[ r_f(\U)=\a_k/\xi(\U);\qquad \l(\U,r)=\r_k(\xi(\U) r).\]

In Appendix~\ref{asc}, we describe the models that we use to analyze $\G(F,\s)$ and the *-core of $\G(F,\s)$.

\section{Unfrozen variables outside of the *-core}\label{sec:unstable}


 Let $x$ be a vertex of $\G(F,\s)$ which is not in the *-core of $\G(F,\s)$.
We will consider how $x$ can be removed during the peeling process used to find the *-core
of $\G(F,\s)$.
More specifically, we consider a sequence of vertices, culminating in $x$,
 which could be removed in sequence by the peeling process.

\begin{definition}
A {\em peeling chain} for a vertex $x\in \G(F,\s)$ is a sequence of vertices $x_1,...,x_{\ell}=x$ such that
each $x_i$ is not essential for any hyperedges in the hypergraph remaining after removing $x_1,...,x_{i-1}$ from $\G(F,\s)$.  The {\em depth} of the chain is the maximum distance from one of the vertices to $x$. The {\em *-depth} of $x$ is the minimum depth over all  peeling chains for $x$.
\end{definition}

In Appendix~\ref{a3}, we will prove:

\begin{lemma} \label{ldepth} { Consider any non-trivial, feasible, symmetric,  balance-dominated, and 1-essential CSP-model $\U$. Let $(F,\s)$ be drawn from the planted model $P(\U,n,M=rn)$ where $r\neq r_f(\U)$.}

For any $\e>0$, there exists constant $L$ such that: For all $g(n)=o(n)$, 
the probability that at least $\e n$ vertices of $\G(F,\s)$
that are not in the *-core of $\G(F,\s)$ have *-depth greater than $L$ is less than $e^{-g(n)}$.
\end{lemma}

This is enough to prove that all but $o(n)$ variables outside the *-core are 1-frozen as follows:

{
{\em Proof outline of Theorem~\ref{mt}(a.ii,b):}

Consider any $\e>0$.  If $(F,\s)$ is drawn from the planted model then, by Lemma~\ref{ldepth}, $\G(F,\s)$ has fewer than $\e n$ vertices of *-depth greater than $L$ 
with probability at least $1-e^{-g(n)}$.  So for $r<r_p(\U)$, Lemma~\ref{tptp} yields that the same is true w.h.p.\ when $(F,\s)$ is drawn from the uniform model.}

 Consider any vertex $x$ of *-depth at most $L$.  Consider a peeling chain for $x$
of depth at most $L$ and let $W$ be the set of all hyperedges  that contain at least one vertex of the peeling chain.  

If no hyperedges of $W$ form a cycle, then it is easy to see that we can change all of the variables in the peeling chain, one-at-a-time and still have a satisfying assignment for $F$.  Indeed, this follows from a straightforward induction on $L$. Therefore, the variable $x$ is not $1$-frozen. The case where $W$ contains a cycle is rare enough to be negligible (see  {Appendix~\ref{scs}}).  So for
all $\e>0$ there are fewer than $\e n$ variables outside of the *-core that are not 1-frozen,
as required. 
\proofend

This argument also leads to:

{\em Proof {outline} of Theorem~\ref{mt}:} This theorem follows as above, by adding the observation
that with sufficiently high probability, almost all vertices outside the *-core have a peeling chain
of size $O(1)$.  We can change the corresponding variable by changing a subset of the entire peeling chain.  See Appendix~\ref{asc} for the short proof.
\proofend

\section{Frozen variables in the *-core}\label{schop}
Most of the work in this paper is in proving that almost all vertices in the *-core of $\G(F,\s)$ are $\Theta(n)$-frozen. To do so, we first study the structure of sets of variables that can be changed to obtain a new solution.  Note that if changing the value of every variable of $S$ yields a solution, then every constraint whose essential variable is in $S$ must contain at least one other variable in $S$. This leads us to define:

 \begin{definition} A {\em flippable set} of the *-core of $\G(F,\s)$ is a set of vertices $S$
such that for every $x\in S$ and every *-core hyperedge $f$ in which $x$ is essential,
$S$ contains another vertex of $f$.
\end{definition}

For every vertex $x\in S$, since $x$ is in the *-core, there will be at least one such hyperedge $f$.  

{Note:  if $S$ is a flippable set in $\G(F,\s)$, then changing the variables of $F$ corresponding to $S$ will not necessarily yield another solution; this will depend on the
actual constraints of $F$.  But it is easily seen that the converse holds:

\begin{proposition} \label{pflip}  If $\s,\s'$ are two solutions to a 1-essential CSP, $F$, then the set of *-core variables
on which they differ form a flippable set in $\G(F,\s)$.
\end{proposition}

\begin{proof} Let $S$ be the set of variables in the *-core of $(F,\s)$ on which $\s,\s'$ differ.
Suppose that $S$ does not form a flippable set in $\G(F,\s)$.  Then there is a variable $x\in S$ and a *-core hyperedge $e$ in which $x$ is essential, such that $e$ contains no other members of $S$.  The hyperedge $e$ corresponds to a constraint in $F$.  In that constrain, the solutions $\s,\s'$ agree on all variables but $x$, which contradicts the fact that $x$ is essential for $e$.
\end{proof}
}

We prove that for some $\phi'(n)=o(n)$ and constant $\z>0$, with sufficiently high probabilty,
there are no flippable sets of size $\phi'(n)<a<\z n$.  This will be enough to prove that at most $o(n)$ vertices lie in flippable sets, which in turn will be enough to show that almost all of the *-core is frozen. 

We apply the first moment method.  Unfortunately, we cannot apply it directly to the number of flippable sets because the existence of one flippable set $S$ typically leads to the existence
of an exponential number of flippable sets formed by adding to $S$ vertices $x$ such that (i) $x$ is essential in exactly one hyperedge, and (ii) that hyperedge contains a non-essential vertex in $S$.  So instead we focus on something that we call {\em weakly flippable sets}, which do not contain such vertices $x$.  Roughly speaking: every flippable set can be formed from a weakly flippable set
by repeatedly adding vertices $x$ in that manner.  We prove that with sufficently high probability:
\begin{enumerate}
\item[(a)] There are no weakly flippable sets of size $\phi(n)<a<\z n$.
\item[(b)] There are no weakly flippable sets of size at most $\phi(n)$  which extend to
a flippable set of size greater than $\phi'(n)$.
\end{enumerate}
This establishes our bound on the sizes of flippable sets. (This is not quite true - we also need to consider {\em cyclic sets} - but it provides a good intuition.)

Let $H_1$ denote the vertices that are essential in exactly one hyperedge.  Define a {\em one-path} to be a sequence of vertices $x_1,...,x_{t+1}$ such that for each $1\leq i\leq t$:
$x_i\in H_1$ and $x_{i+1}$ is in the hyperedge in which $x_i$ is essential.  Note that if $x_{t+1}$ is in a  flippable set $S$, then we can add the entire one-path to $S$ and it will still be flippable.
This ends up implying that if we have a proliferation of long one-paths, then we would not be able to prove (b).  It turns out that a proliferation of long one-paths would also prevent us from proving (a).

Consider a vertex $x\in H_1$ and the edge $f$ in which $x$ is essential.  Intutively, the expected number of other members of $H_1$ that are in $f$ is $(k-1)|H_1|$ divided by the
size of the *-core.  We prove (Lemma~\ref{lbranch}) that this ratio is less than 1.  This implies that one-paths do not ``branch'' and so we do not tend to get many long one-paths. So our bound on this ratio plays a key role in establishing both (a) and (b).

This is just an intuition.  In fact, {\em one-paths} are not explicitly mentioned anywhere in the proofs.  For all the details, see Appendix~\ref{sec:stable}.

\section{Further Challenges}\label{sec:conclusion}
Of course, one ongoing challenge is to continue to rigorously establish parts of the clustering picture.  By now, it is clear that in order to establish satisfiability thresholds or understand the algorithmic challenges for problems with densities approaching that threshold, we will {probably} need a strong understanding of clustering.

Another challenge is to try to establish whether the freezing threshold is, indeed, an algorithmic barrier.  For several CSP-models, we now know the precise location of that threshold, and we have a very good understanding of how it arises and which variables are frozen.  Perhaps we can use that understanding to prove that a simple algorithm works for all densities up to that threshold and/or establish that frozen clusters will indeed neccesitate more sophistication.

Another challenge is to  determine the freezing threshold for a wider
variety of CSP-models.  These techniques rely crucially on {using} the planted model; at this point
there is no known way to get to the exact threshold without it.  This prevents us from
extending our results to $k$-SAT and many other models as the planted model does not work nearly well enough,
mainly because the number of solutions is not sufficiently concentrated.  A more important challenge
would be to devise a better means to analyze random solutions to CSP's drawn from those models.

\section*{Acknowledgment}

The authors are supported by an NSERC Discovery Grant and an NSERC Accelerator Fund.

\newpage

\begin{center}
{\Large{\bf Appendix}}
\end{center}

\section{The transfer theorem}\label{a1}



{Let us consider a CSP-model $\U=(\Phi,p)$. Let us recall the properties from Definition \ref{def:boolprop}.} Given a boolean function $\varphi\in\Phi$, {we} denote by $S_{\varphi}$ 
the set of satisfying assignments of $\varphi$ {and also we define $I_{\varphi}:=\{-1,1\}^k\setminus S_{\varphi}$}. Now, let $\varphi(x)=\sum\limits_{Q\subseteq\{-1,1\}^k}\left(\varphi_Q\prod_{i\in Q}x_i\right)$ be its \emph{Fourier expansion}. Such expansion is unique with $\varphi_Q:=\sum\limits_{x\in\{-1,1\}^k}\left(\varphi(x)\prod_{i\in Q}x_i\right)$. In particular, {it is the case} that $\varphi_{\emptyset}=\frac{|S_{\varphi}|}{2^k}=\sum\limits_{Q\subseteq\{-1,1\}^k}\varphi_Q^2$. {Moreover, if $\varphi$ is symmetric, we have that} $\varphi_{\{i\}}=0$ (In fact, $\varphi_Q=0$ whenever $|Q|$ is odd). Now, we define the polynomial $p_{\varphi}(\theta)$ as follows,
\[
p_{\varphi}(\theta):=\sum_{Q\subseteq \{-1,1\}^k}(\varphi_Q/\varphi_{\emptyset})^2\theta^{|Q|}
\]
Also, we define the \emph{binary entropy function} $\Hh(\theta)$ as  
\[
\Hh(\theta):=-\frac{1+\theta}{2}\ln(1+\theta)-\frac{1-\theta}{2}\ln(1-\theta)
\]
{Finally, we define}
\[
r_p(\U):=\inf_{\theta\in(0,1)}\frac{-\Hh(\theta)}{\sum_{\varphi\in\Phi}p_{\varphi}\ln(p_{\varphi}(\theta))}.
\]

We will now prove Lemma~\ref{tptp}, which we restate:

\noindent {\bf Lemma~\ref{tptp} } {\em
  Consider any non-trivial, feasible, symmetric, and balance-dominated CSP-model $\U$.  
For  every $r<r_p(\U)$, there is
a function $g(n)=o(n)$ such that: Let $\cale$ be any property of pairs $(F,\s)$ where $\s$ is a satisfying solution of $F$. 
If
\[\pr(P(\U,n,M=rn) \mbox{ has } \cale)>1-e^{-g(n)},\]
then
\[\pr(U(\U,n,M=rn) \mbox{ has } \cale)>1-o(1).\]
}

The proof follows the argument employed in \cite{mrt} to prove Theorem B.3, which followed the same spirit of similar results in \cite{ac}.

\begin{proof}
In what follows, we will take expectations over a random $\vp$ chosen from $\Phi$ with distribution $p$.  Thus, for a variable $X(\vp)$,  we have $\expect(X)=\sum_{\vp\in\Phi}p(\vp)X(\vp)$.
Let $\xi_{\varphi}$ be the number of clauses with constraint $\varphi$ in the random CSP $H$ drawn from $\U$. {Let $\gamma$ be a fixed constant in $(0,1/2)$ and let $\mathcal{F}$ be the event `For all $\varphi\in\Phi$, $|\xi_{\varphi}-\alpha p_{\varphi} n|<n^{1/2+\gamma}$'}. So, $\mathcal{F}$ holds w.h.p.

{We say that a solution $\s$ is {\em balanced} if the number of variables assigned +1 is either $\lceil\frac{n}{2}\rceil$ or  $\lfloor\frac{n}{2}\rfloor$.}
Let $Z_b$ be the number of balanced solutions of $H$, let $Z$ be the number of solutions of $H$ and let $Z_b(\theta)$ be the number of \emph{pairs} of balanced solutions $x^{(1)}$,$x^{(2)}$ with discrepancy $\theta$, that is, such that $\frac{1}{n}\sum_{i=1}^n x^{(1)}_i x^{(2)}_i=\theta$. Now,
\[
\frac{\expect[Z_b^2\I(\mathcal{F})]}{(\expect[Z_b\I(\mathcal{F})])^2}=\sum_{\theta\in U_n}\frac{\expect[Z_b(\theta)\I(\mathcal{F})]}{(\expect[Z_b\I(\mathcal{F})])^2}
\]
where $U_n:=\{i/n:i=-n,\ldots,n\}$. From lemma A.2 in \cite{mrt}, {then} it is the case that
\[
\frac{\expect[Z_b(\theta)\I(\mathcal{F})]}{(\expect[Z_b\I(\mathcal{F})])^2}\leq C n^{-1/2}\exp\left(n(\Hh(\theta)+\alpha\expect[\ln(p_{\varphi}(\theta))]+o(1))\right)
\]
where $C$ does not depends on $\theta$ (neither the $o(1)$ term). 

Now, if
$\alpha<r_p(\U)$, it is the case that 
\begin{equation}\label{eq:1}
H(\theta)+\alpha\expect[\ln(p_{\varphi}(\theta))]<0 \text{ for all } \theta\in(0,1).
\end{equation}

On the other hand, {since $\U$ is symmetric}, 
\[
H(\theta)+\alpha\expect[\ln(p_{\varphi}(\theta))]=\left(-\frac{1}{2}+\alpha\expect\left[\sum_{|Q|=2}(\varphi_Q/\varphi_{\emptyset})^2\right]\right)\theta^2+O(\theta^4).
\]
{Now, since}
\[
\lim_{\theta\to 0}\frac{-H(\theta)}{\expect[\ln(p_{\varphi}(\theta))]}=\frac{1/2}{\expect\left[\sum_{|Q|=2}(\varphi_Q/\varphi_{\emptyset})^2\right]}>\alpha,
\]
then it is the case 
$H(\theta)+\alpha\expect[\ln(p_{\varphi}(\theta))]<-c\theta^2$ for some $c>0$ and $\theta$ close enough to $0$. Combining this fact with eq. \eqref{eq:1}, we have that for some $c'>0$,
\begin{equation}\label{eq:12}
H(\theta)+\alpha\expect[\ln(p_{\varphi}(\theta))]<-c'\theta^2 \text{ for all } \theta\in(0,1).
\end{equation}
Now,
\begin{eqnarray}
\frac{\expect[Z_b^2\I(\mathcal{F})]}{(\expect[Z_b\I(\mathcal{F})])^2}&\leq&\frac{C}{n^{1/2}}\sum_{\theta\in U_n} \exp(-c'n(\theta^2+o(1)))\\
&\leq& Cn^{1/2}\int_{-\infty}^{\infty}\exp(-c'n(\theta^2+o(1)))
\end{eqnarray}
And the last quantity is bounded by a constant $C_0$ (not depending on $n$). {This implies, by the Paley-Zygmund inequality, that} for every $\epsilon>0$ and all $n\geq n_0$ it is the case that $\prob(Z_b>e^{-n\epsilon}\expect[Z_b])\geq C_0/2$.

Now, {because $\U$ is balance-dominated}, we have that $\expect[Z]\leq n\expect[Z_b]$. Therefore, for $n$ large enough, we have that
\[
\prob(Z>e^{-n\epsilon}\expect[Z])\geq\prob(Z_b>n e^{-n\epsilon}\expect[Z_b])\geq\prob(Z_b\geq e^{-n(\epsilon/2)}\expect[Z_b])\geq C_0/2.
\]
{On the other hand,} it is easy to see that $\expect[Z]$ is exponential in $n$ for $\alpha<r_p(\U)$ (Indeed $\expect[Z]$ is exponential for $\alpha<r_{sat}(\U):=\frac{\ln 2}{\expect_{\varphi}[\ln(1+|I_{\varphi}|/|S_{\varphi}|)]}=\frac{-H(1)}{\expect[\ln(p_{\varphi}(1))]}$). {Now, let us recall from Appendix C in \cite{mrt}, that the event `$Z>B^n$', where $B>1$, has a sharp threshold in the clauses to variables ratio. Thus, the event `$Z>e^{-n\epsilon}\expect[Z]$' has a sharp threshold in the parameter $\alpha$}. Therefore, necessarily, it is the case {that $Z>e^{-n\epsilon}\expect[Z]$ w.h.p.}. This implies therefore, that for some function $g(n)$ of order $o(n)$, it is the case that w.h.p., 
\begin{equation}\label{eq:transfer}
\ln(Z)>\ln(\expect(Z))-g(n).
\end{equation}
After this equation is established now the lemma follows. For instance, from Theorem B.3 in \cite{mrt}.
\end{proof}
 
Now, recall our other property:

{\bf 1-essential:} {\em 
Given a boolean constraint $\varphi$ and an assignment $\sigma$  that satisfies
$\varphi$, we say that the variable $x$ is {\em  essential} for $(\varphi,\sigma)$ if changing the value of $x$ results in $\varphi$ being unsatisfied.
We say that a set $\Phi$ of constraints is {\em 1-essential} if for every  $\varphi\in\Phi$, and every $\s$ satisfying $\varphi$,
at most one variable is essential for $(\varphi,\s)$.  A CSP-model $(\Phi,p)$ is  {\em 1-essential} if $\Phi$ is   1-essential.
}

An easy description of a feasible, 1-essential constraint is the following: $\varphi$ is feasible and 1-essential iff the Hamming distance between any pair of assignments in $I_{\varphi}$ is greater than $2$. This implies in particular that $|I_{\varphi}|\leq \frac{2^k}{\binom{k}{2}+1}$ {and} $\varphi_{\{i,j\}}=-\frac{1}{2^k}\sum_{x\in I_{\varphi}}x_i x_j$. This allows us to prove a more concrete lower bound on the transfer threshold $r_p(\U)$ that we will use in the next section to establish that $r_p(\U)$ is above the freezing threshold for large enough $k$.

\begin{theorem}~\label{tomega}
{Consider any non-trivial, feasible, symmetric, balance-dominated and 1-essential CSP-model $\U$. It is the case that 
\[
r_p(\U)\geq \frac{0.25}{\Omega_p(\U)},
\]
where
\[
\Omega_p(\U):=\expect_{\varphi}[|I_{\varphi}|/|S_{\varphi}|].
\]
}
\end{theorem}
\begin{proof}
Since every constraint $\varphi\in\Phi$ is feasible and 1-essential, we have that
\begin{eqnarray*}
\sum_{\{i,j\}}\left(\frac{\varphi_{\{i,j\}}}{\varphi_{\emptyset}}\right)^2
=\sum_{\{i,j\}}\frac{\left(\sum_{x\in I_{\varphi}} x_i x_j\right)^2}{|S_{\varphi}|^2}\leq \binom{k}{2}\left(\frac{|I_{\varphi}|}{|S_{\varphi}|}\right)^2
\end{eqnarray*}
Therefore, since 
\begin{eqnarray*}
\sum_{|Q|\geq4}\varphi_Q^2\theta^{|Q|}
\leq\sum_{|Q|\geq4}\varphi_Q^2\theta^{4}
\leq\left(\sum_{Q\subseteq\{-1,1\}^k}\varphi_Q^2-\varphi_{\emptyset}^2\right)\theta^{4}
=\varphi_{\emptyset}(1-\varphi_{\emptyset})\theta^4,
\end{eqnarray*}
we have that 
\[
p_{\varphi(\theta)}\leq 1+\binom{k}{2}\left(\frac{|I_{\varphi}|}{|S_{\varphi}|}\right)^2 \theta^2+\frac{|I_{\varphi}|}{|S_{\varphi}|}\theta^4
\]
And, since $|I_{\varphi}|\leq\frac{2^k}{\binom{k}{2}+1}$, and therefore $\binom{k}{2}\left(\frac{|I_{\varphi}|}{|S_{\varphi}|}\right)^2\leq \frac{|I_{\varphi}|}{|S_{\varphi}|}$, we get that
\[
p_{\varphi(\theta)}\leq 1+2\frac{|I_{\varphi}|}{|S_{\varphi}|}\theta^2
\]
Thus,
\[
\expect_{\varphi}[\ln(p_{\varphi}(\theta))]\leq 2\theta^2\Omega_p(\U)
\]
Now, we finally conclude that
\begin{eqnarray}
r_p(\U)=\inf_{\theta\in(0,1)}\frac{-H(\theta)}{\expect_{\varphi}[\ln(p_{\varphi}(\theta))]}\geq \frac{0.5}{\Omega_p(\U)} \inf_{\theta\in(0,1)}\frac{-H(\theta)}{\theta^2}=\frac{0.25}{\Omega_p(\U)}.
\end{eqnarray}
\end{proof}

We close this section by discussing the CSP-models that satisfy our five conditions: non-trivial, feasible, symmetric,  balance-dominated, and 1-essential.

Our properties are rich enough to permit a large class of CSP-models beyond hypergraph 2-coloring and $k$-NAE-SAT. 
For example, we can construct a model in the following way: 

Represent the assignments in $\{-1,+1\}^k$ as the $k$-dimensional hypercube $H_k$,
and so two assignments are adjacent if they differ on exactly one variable. Let $L_{\e}$ denote the vertices $x \in H_k$ with $\sum x_k>\epsilon k$. Consider any subset $I\subseteq L_{\e}$  containing no two vertices of distance at most two.  We use $-I$ to denote the subset formed
by switching the sign of every vertex in $I$, and set $J:=I \cup -I$ to be the assignments which violate our constraint $\vp_J$. I.e., $\varphi_J(x):=1$ iff $x\notin J$. 

Now consider any set $\Phi$ of constraints of this form in which at least one is non-trivial
(i.e. has $(1,1,...,1)\in J$). Let $\U=(\Phi,p)$ for any $p$ (such that $\supp(p)=\Phi$). For \emph{any $k$ large enough in terms of $\e$}, $\Phi$ satisfies our five properties.
For instance, hypergraph 2-coloring is formed in this way with  $I:=(1,...1).$

Given a constraint $\vp$ and some $s\in\{-1,+1\}^k$, we define the constraint $\vp^s$ as
$\vp^s(x_1,...,x_k)=\vp(s_1x_1,...,s_kx_k)$. We can allow $\e=0$ and drop the condition
that $k$ must be large if (a) no two vertices of $J$ are within distance 2, and (b) for every $\vp\in\Phi$ and every $s\in\{-1,+1\}^k$, we have $\vp^s\in\Phi$ and $p(\vp^s)=p(\vp)$.
For instance, $k$-NAE-SAT is formed in this way with $I:=(1,...,1)$.

\section{Essential hyperedges}\label{a2}

Consider any nontrivial, feasible, symmetric 1-essential CSP-model $\U=(\Phi,p)$.  We will draw  $(F,\s)$ from the planted model  $P(\U,n,M)$. We begin by taking a random assignment $\s$ for the variables $x_1,...,x_n$ and note that $|\La^+|,|\La^-|=\hf n+o(n)$ with probability at least $1-e^{-g(n)}$, for any $g(n)=o(n)$.  So we can assume that this condition holds.

{In what follows, we will take expectations over a random $\vp$ chosen from $\Phi$ with distribution $p$.  Thus, for a variable $X(\vp)$,  we have $\expect(X)=\sum_{\vp\in\Phi}p(\vp)X(\vp)$.}

For every $\vp\in\Phi$, recall from the previous section that $S_{\vp}$ is the set of assigments in $\{-1,+1\}^k$ that
satisfy $\vp$ and $I_{\vp}=\overline{S_{\vp}}$ is the set that do not satisfy $\vp$.  We define
$S^e_{\vp}\subseteq S_{\vp}$ to be the set of assignments that satisfy $\vp$ and for which
$\vp$ has an essential variable.  Noting that switching the essential variable of an assignment
in $S^e_{\vp}$ yields an assignment in $I_{\vp}$, and using the fact that $\U$ is feasible, it is easy to see that $|S^e_{\vp}|=k|I_{\vp}|$.

Since $|\La^+|,|\La^-|=\hf n+o(n)$, it follows that when picking a constraint in
the planted model, we choose $\vp$ with probability proportional to $p(\vp)|S_{\vp}|+o(1)$.
  Thus, 
defining $\Omega_f:=\frac{\expect|I_{\varphi}|}{\expect|S_{\varphi}|}$, the probability that
$\vp$ has an essential variable is: 

\[\xi(\U)=k\Omega_f+o(1).\]

So the number of constraints that have an essential variable is distributed as the binomial $BIN(M=r n,\xi(\U))$.   Concentration of the binomial variable implies Lemma~\ref{lxi}.

Now recall the type of $\vp$, as defined in Section~\ref{shm}.  
For a constraint $\varphi\in\Phi$, define $I_{\varphi}(a,b):=\{x\in I_{\varphi}: x \text{ has } a \text{ } 1's \text{ and } b \text{ } -1's\}$ then the clause $\varphi$ has exactly $(b+1)|I_{\varphi}(a,b+1)|$ assignments of type $(1;a,b)$ and $(a+1)|I_{\varphi}(a+1,b)|$ assignments of type $(-1;a,b)$. Therefore, when picking a constraint in
the planted model, if we condition on the event that it has an essential variable, then the conditional probability that it has type $\tau=(1;a,b)$ is 
\[
\gamma_{\tau}=\frac{(b+1)\expect[|I_{\varphi}(a,b+1)|]}{k\expect[|I_{\varphi}|]}+o(1)
\]
and to be of type $\tau=(-1;a,b)$ is 
\[
\gamma_{\tau}=\frac{(a+1)\expect[|I_{\varphi}(a+1,b)|]}{k\expect[|I_{\varphi}|]}+o(1)
\]
Since $\U$ is symmetric, $\vp(x)=\vp(-x)$ for every assignment $x$. It follows that $|I_{\varphi}(a,b)|=|I_{\varphi}(b,a)|$ and therefore $\gamma_{\tau=(1;a,b)}=\gamma_{\tau=(-1;b,a)}+o(1)$. 
So, noting that we can exchange $a,b$ in the following definition: 
\[
\gamma^+:=\sum_{\tau=(1,a,b)} \gamma_{\tau} \text{, }\qquad \gamma^-:=\sum_{\tau=(-1,a,b)} \gamma_{\tau},
\]
we have $\gamma^+=\gamma^-=\hf+o(1)$. In other words:
\begin{lemma}\label{lxx}
When we choose a random clause for the planted model, and condition on it having an essential variable: the probability that the essential
variable is in $\La^+$ is equal to the probability that it is in $\La^-$ plus $o(1)$.
\end{lemma}

We close this section by showing that $r_f(\U)<r_p(\U)$ for sufficiently large $k$.

\begin{proposition}\label{p27}
For any nontrivial, symmetric, feasible, balance-dominated, 1 essential CSP model $\U$ of arity $k$:
\begin{enumerate}
\item[(a)] For every $k\geq 27$, $r_p(\U)>r_f(\U)$.
\item[(b)] Asymptotically in $k$, $\frac{r_f(\U)}{r_p(\U)}\lesssim \frac{\ln k}{k}$.
\end{enumerate}
\end{proposition}

\begin{proof}
Notice first that 
\begin{eqnarray*}
\Omega_p = \expect\left[\frac{|I_{\varphi}|}{|S_{\varphi}|}\right]
\leq \frac{\expect[|I_{\varphi}|]}{2^k(1-\frac{1}{\binom{k}{2}+1})}
\leq  \frac{\expect[|I_{\varphi}|]}{(1-\frac{1}{\binom{k}{2}+1})\expect[|S_{\varphi}|]}
=\frac{\Omega_f}{(1-\frac{1}{\binom{k}{2}+1})}.
\end{eqnarray*}
Notice also that $\alpha_k\leq \frac{2\ln(k)}{(1-1/k^2)^{k-1}}$. Therefore, since 
\[
\frac{2\ln(k)}{ k(1-1/k^2)^{k-1}}\leq (1/4)(1-\frac{1}{\binom{k}{2}+1})
\]
for $k\geq 27$, then
\[
r_f(\U)\leq \frac{2\ln(k)}{\Omega_f k(1-1/k^2)^{k-1}}\leq \frac{(1/4)}{\Omega_p}\leq r_p(\U), 
\]
{by Theorem~\ref{tomega}.}
Then, part (a) follows. To prove part (b) {we use the previous inequality, so that}
\[
\frac{r_f(\U)}{r_p(\U)}\leq \frac{8\ln(k)}{k(1-\frac{1}{\binom{k}{2}+1})(1-1/k^2)^{k-1}}\sim 8 \ln(k)/k
\]
\end{proof}

\section{The *-core}\label{asc}

\begin{lemma}\label{lstarfr} Consider any 1-essential CSP-model $\U=(\Phi,p)$ of arity $k$,
and a random CSP, $F$, drawn from $P(\U,n,M=rn)$. Suppose $\G(F,\s)$ has $\a n +o(n)$ hyperedges with $\a\neq\a_k$. For any $g(n)=o(n)$ and constant $\e>0$, there exist constants $T,Z,\b>0$ such that, with probability at least $1-e^{-g(n)}$:
\begin{enumerate}
\item[(a)] All but $o(n)$ vertices of the *-core of  $\G(F,\s)$ are $\b n$-frozen variables for $(F,\s)$.
\item[(b)] All but at most $\e n$ vertices outside the *-core of  $\G(F,\s)$ are either (i) not $T$-frozen variables for $(F,\s)$ or (ii) within distance $Z$ from a cycle of length at most $Z$.
\end{enumerate}
\end{lemma}

This  yields  Theorem~\ref{mt}:

{\em Proof of Theorem~\ref{mt}:} Consider $(F,\s)$ drawn from the uniform model $U(\U,n,M=rn)$.  A simple first moment calculation shows that the expected number of variables  that are within distance $Z$ of a cycle of length at most $Z$ in the underlying hypergraph of $F$ is $O(1)$.  Therefore w.h.p. there are $o(n)$ such vertices.

For part (b): If $r>r_f(\U)$ then $\a>\a_k$.  Consider any $\e>0$.  Lemma~\ref{tptp} allows us to transfer  Lemmas~\ref{lHp},~\ref{lstarfr},~\ref{lxi} to $(F,\s)$ to establish that w.h.p.\ all but at most $\e n$ variables are
either $T$-frozen with respect to $\s$ or are within distance $Z$ of a cycle of length at most $Z$.  W.h.p.\ there are $o(n)$ variables of the latter type, and so all but at most $\e n+o(n)$
vertices are $T$-frozen.  By letting $T$ tend to infinity we can take $\e$ arbitrarily small
thus obtaining part (b).

For part (a):  If $r>r_f(\U)$ then $\a<\a_k$. Again, we transfer Lemmas~\ref{lHp},~\ref{lstarfr},~\ref{lxi} to $(F,\s)$. This shows that w.h.p. all but $o(n)$ of the vertices of the *-core are frozen.  The same
argument as for part (b) shows that w.h.p. all but $o(n)$ of the vertices outside of the *-core
are frozen. Part (a) follows since $\l(\U,r)=\r_k(\xi(\U) r)=\r_k(\a)$ and w.h.p.\ the size of the *-core
is $\r_k(\a)n+o(n)$.
\proofend

Lemma~\ref{lstarfr}(a) is proven in Section~\ref{sec:stable}. Lemma~\ref{lstarfr}(b) follows
from Lemma~\ref{ldepth} as follows:

{\em Proof of Lemma~\ref{lstarfr}(b):} Consider any $\e>0$. 
With probability at least $1-e^{-g(n)}$, $\G(F,\s)$ has fewer than $\e n$ vertices of *-depth greater than $L$, where $L$ comes from Lemma~\ref{ldepth}.  Consider any vertex $x$ of *-depth at most $L$.  Consider a peeling chain for $x$
of depth at most $L$ and let $W$ be the set of all hyperedges  that contain at least one vertex of the peeling chain.  

If some hyperedges of $W$ form a cycle, then there must be a cycle of length at most $2L$ within distance $L$ of $x$.   
If no hyperedges of $W$ form a cycle, then it is easy to see that we can change all of the variables in the peeling chain, one-at-a-time and still have a satisfying assignment for $F$.  Indeed, this follows from a straightforward induction on $L$. Therefore, the variable $x$ is not $1$-frozen.
\proofend

\subsection{Our hypergraph models}\label{shm}

 Consider any 1-essential CSP, $F$, and any solution $\s$.

The vertices of $\G(F,\s)$ are partitioned into two sets $\La^+,\La^-$ containing those variables which are assigned $+1,-1$ respectively under $\s$.

\begin{definition} For each hyperedge $e\in \G(F,\s)$:
Let $a$ be the number of non-essential vertices of $e$ in $\La^+$ and let $b$ be the number of non-essential vertices of $e$ in $\La^-$.  The {\em type} of $e$ is defined to be:
\begin{itemize}
\item $(1,a,b)$ if  the essential vertex vertex of $e$ is in $\La^+$;
\item $(-1,a,b)$, if  the essential vertex vertex of $e$ is in $\La^-$.
\end{itemize}
The {\em type} of a constraint of $(F,\s)$ with an essential vertex, is the
type of the corresponding hyperedge in $\G(F,\s)$.
\end{definition}

Now consider a nontrivial, feasible, symmetric, balance-dominated, 1-essential  CSP-model $\U$ and choose a random $(F,\s)$
from the planted model $P(\U,n,M)$.  
Recalling the Remark following Definition~\ref{dplant}, we can selected the constraints of $F$ independently. Given the partition $\La^+,\La^-$,
and a type $\t$, we let $w(\t)=w(\t,\La^+,\La^-)$ denote the probability that a selected constraint has type $\t$,
conditional on it having an essential vertex. (See Appendix~\ref{a2} for further discussion.)  Note that $w(\t)$ depends only on $\U,|\La^+|,|\La^-|$.   Note further that, conditional on a hyperedge $e$ having type $\t$, every choice of the
vertices of $e$ which is consistent with $\t$ is equally likely.  Thus, when choosing $\G(F,\s)$
we can choose the type of a hyperedge first and then its vertices.  This leads us to:

\noindent{\bf Model A:}
\begin{enumerate}
\item Partition the vertices into $\La^+,\La^-$ uniformly at random.
\item For $i=1$ to $M$, choose the $i$th hyperedge $e_i$ as follows:
\begin{enumerate}
\item Choose the type $(s,a,b)$ of $e_i$ (where $s\in\{+1,-1\}$), where  type $\t$ is chosen
with probability $w(\t)$.
\item Choose the essential vertex for $e_i$ uniformly from the appropriate set, $\La^+$ or $\La^-$, according to $s$.
\item Choose $a$ vertices uniformly from $\La^+$ and $b$ vertices uniformly from $\La^-$.
These are the non-essential vertices of $e_i$.
\end{enumerate} 
\end{enumerate}

In some cases, it will be useful to fix the essential vertex of every hyperedge, along with the assignment $\s$, and then choose our planted hypergraph.  In this case, for $s\in\{-1,+1\}$, we use $w^s(\t)=w(\t,\La^+,\La^-)$ denote the probability that a selected constraint has type $\t$,
conditional on it having an essential vertex in $\La^s$.
We can use the following model:

\noindent{\bf The Essential Model:}
\begin{enumerate}
\item We are given a partition the vertices into $\La^+,\La^-$.
\item For $i=1$ to $M$, we are given the essential vertex of $e_i$.
We choose the rest of $e_i$ as follows:
\begin{enumerate}
\item Choose the type $(s,a,b)$ of $e_i$, where  $s$ is already determined and type $\t$ is chosen
with probability $w^s(\t)$.
\item Choose $a$ vertices uniformly from $\La^+$ and $b$ vertices uniformly from $\La^-$.
These are the non-essential vertices of $e_i$.
\end{enumerate} 
\end{enumerate}

The essential model will be useful in analyzing the *-core of $\G(F,\s)$.

We let $H_1$ denote the set of vertices $v\in H^*$  that
are essential in exactly one hyperedge.   We use $H_1^+,H_1^-$ to denote $H_1\cap \La^+,H_1\cap \La^-$, the vertices of $H_1$ corresponding to variables assigned $+1,-1$ by $\s$.
The following  lemma will be key in proving that
most of $H^*$ is frozen:

\begin{lemma}\label{lbranch}  If $\U$ is non-trivial, feasible, symmetric, and balance-dominated and if $\a>\a_k$ then there exists $\g=\g(\U,\a)>0$ such that:
for any $g(n)=o(n)$, with probability at least $1-e^{-g(n)}$,
\begin{enumerate}
\item[(a)] $|V(H^*)\cap \La^+|,|V(H^*)\cap \La^-| =|V(H^*)|(\hf+o(1))$;
\item[(b)] $|H^+_1|,| H^-_1|\leq \frac{\hf-\g}{k-1}|V(H^*)|.$
\end{enumerate}
\end{lemma}

The proof appears in Appendix~\ref{a3}.

We close this section with:

{\em Proof of Theorem~\ref{mt}:} Since $r<r_f(\U)$, w.h.p. the *-core is empty.
During the proof of Lemma~\ref{lconc} in Appendix~\ref{a3},
we prove that for $D$ sufficiently large,  with probability at least $1-e^{-g(n)}$, fewer than $\e n$ vertices are within distance $L$ of a vertex with degree greater than $D$.  It follows that for all but at most $\e n$ vertices of depth at most $I$, the size of their peeling chain is at most $(kD)^I=O(1)$.  We can change any such variable by changing a subset
of the entire peeling chain in one step. So, applying Lemma~\ref{ldepth},
we see that for all but  $2\e n$ variables $v$,  we can change $v$ by changing at most $(kD)^I$
variables.  

We use Lemma~\ref{tptp} to show that this holds
w.h.p.\ in the uniform model.  Then by taking $D$ arbitrarily large and $\e$ arbitarily small,we obtain the theorem.
\proofend
\section{Analysis of the *-core process}\label{a3}

Recall that $\U$ is a non-trivial, feasible, symmetric,  balance-dominated, and 1-essential CSP-model, and that we draw $(F,\s)$ from the planted model.

Let $H$ denote the hypergraph $\G(F,\s)$.  $H$ has $M=\a n$ edges.  We will analyze the *-core process on $H$ (recall Section~\ref{sstar}).  We follow the analysis of \cite{mmcore},  being careful to obtain a failure
probability of at most $e^{-g(n)}$ for any $g(n)=o(n)$;  alternatively, we could have followed
the analysis of \cite{jhk}.

Recall that $\La^+,\La^-$ denotes the sets of vertices corresponding to variables of sign $+1,-1$ in $\s$.  We can assume that $|\La^+|,|\La^-|=\hf n+o(n)$, as this occurs with probability $1-e^{-g(n)}$ for any $g(n)=o(n)$.

Let $H(0)=H$ and define $H(i+1)$ to be the hypergraph obtained by removing every vertex in $H(i)$ that is not essential for any hyperedges, along with all hyperedges in which that vertex is non-essential.  We call this operation a {\em parallel round} of the *-core process.  We begin by
analyzing $H(i)$ for constant $i$, using Model A from section~\ref{shm}.

We let $\r^+_i,\r^-_i$ denote the probability that a vertex $v\in\La^+,\La^-$ survives the
$i$ parallel rounds; i.e. $\pr(v\in H(i))$. Initially $\r^+_0=\r^-_0=1$; it will
follow by induction that $\r^+_i=\r^-_i+o(1)$.  So we will recursively define $\r_i$
and show that $\r^+_i=\r^-_i=\r_i+o(1)$.

Consider any vertex $v$.  Note that $v\in H(i+1)$ iff there is at least one hyperedge $f$ in which $v$ is the essential vertex and every non-essential vertex is in $H(i)$.
Lemma~\ref{lxx} implies the following
key property:

\begin{property}\label{pplus}
For every vertex $v$, the expected number of hyperedges in which $v$ is essential is $\a+o(1)$.
\end{property}

Consider any hyperedge $e$ in which $v$ is the essential vertex.  Let the other vertices
be $u_1,...,u_{k-1}$.  By induction, $\pr(u_j\in H(i))=\r_i+o(1)$ for each $1\leq j\leq k-1$. 
W.h.p.\ $F$ is locally tree-like; in particular $v$ does not lie within distance $i$ of a cycle of length at most $2i$.  From this, it is straightforward to show that these  $k-1$ events are nearly independent, and so $\pr(u_1,...,u_{k-1}\in H(i))=\r_i^{k-1}+o(1)$.  Furthermore, Property~\ref{pplus} and the fact that w.h.p. $v$ does not lie near a short cycle imply that the expected number of hyperedges in which $v$ is essential and all non-essential vertices are in $H(i)$ is $\l_i+o(n)$ where
\[\l_i=\a\r_i^{k-1}.\]
A similar straightforward expected number calculation shows that for any $t>0$, the expected number
of $t$-tuples of such hyperedges is $\l_i^t+o(1)$; again, the key point is that if there are no nearby short
cycles, then the hyperedges occur nearly independently.  So the Method of Moments 
(see  e.g.\ \cite{jlr}) implies that the number of such hyperedges is distributed asymptotically as a Poisson.
In particular, the probability that there is at least one is $\r_{i+1}+o(1)$ where
\[\r_{i+1}=1-e^{-\l_{i}}=1-e^{-\a\r_i^{k-1}}.\]
In other words $\r^+_{i+1},\r^-_{i+1}=\r_{i+1}+o(1)$, thus completing the induction.
We define
\begin{itemize}
\item $X^+_i,X^-_i$ is the number of vertices of $\La^+,\La^-$ in $H(i)$;
\item $Y^+_i,Y^-_i$ is the number of hyperedges in $H(i)$ whose essential vertex is in $\La^+,\La^-$;
\item $A^+_i,A^-_i$ is the number of vertices of $\La^+,\La^-$ in $H(i)$ that are not essential in any hyperedges of $H(i)$;
\item $B^+_i,B^-_i$ is the number of vertices of $\La^+,\La^-$  that are essential in exactly one hyperedge of $H(i)$.
\end{itemize}
By the above calculations, $\ex(X^+_i),\ex(X^-_i)=\hf\r_i n+o(n)$.  Since every hyperedge has exactly one
essential variable, those calculations yield $\ex(Y^+_i),\ex(Y^-_i)=\hf\l_i n+o(n)$.   $A^+_i,A^-_i$ count the vertices
that are in $H(i)$ but not in $H(i+1)$; so $\ex(A^+_i),\ex(A^-_i)=\hf(\r_i-\r_{i+1})n+o(n)$.  Since the number
of edges in which $v$ is essential is asymptotic to a Poisson with mean $\l_i$, $\ex(B^+_i),\ex(B^-_i)=\hf\l_ie^{-\l_i}n+o(n)$.  We will prove below that these variables are all highly concentrated.

\begin{lemma}\label{lconc}
For any fixed $i\geq 0$, and any constant $\e>0$, there exists $\eta=\eta(\e,\a,i,\U)$:
\begin{enumerate}
\item[(a)] $\pr(|X^+_i- \hf\r_i n|>\e n)<e^{-\eta n}$, $\pr(|X^-_i- \hf\r_i n|>\e n)<e^{-\eta n}$
\item[(b)] $\pr(|Y^+_i- \hf\l_i n|>\e n)<e^{-\eta n}$, $\pr(|Y^-_i- \hf\l_i n|>\e n)<e^{-\eta n}$;
\item[(c)] $\pr(|A^+_i- \hf(\r_i-\r_{i+1})n|>\e n)<e^{-\eta n}$, $\pr(|A^-_i- \hf(\r_i-\r_{i+1})n|>\e n)<e^{-\eta n}$;
\item[(d)] $\pr(|B^+_i- \hf\l_ie^{-\l_i} n|>\e n)<e^{-\eta n}$,$\pr(|B^-_i- \hf\l_ie^{-\l_i} n|>\e n)<e^{-\eta n}$.
\end{enumerate}
\end{lemma}
We defer the proof to the end of this appendix.

We let $\r=\lim_{i\rightarrow\infty}\r_i$, which exists since $\r_i$ is positive and decreasing.  So $\r$ must satisfy $\r=1-e^{-\a\r^{k-1}}$.  Setting $\l=\lim_{i\rightarrow\infty}\l_i=\a\r^{k-1}$, we obtain:
\[\r=1-e^{-\l};\qquad\text{ so } \l=\a(1-e^{-\l})^{k-1};
\qquad \text{ so } \a=\frac{\l}{(1-e^{-\l})^{k-1}}.\]

We will prove:
\begin{lemma}\label{lcore}
 For any $g(n)=o(n)$,  with probability at least $1-e^{-g(n)}$:
\begin{enumerate}
\item[(a)]  If $\a< \a_k$ then the *-core of  $\G(F,\s)$ has $o(n)$ vertices.
\item[(b)  ]If $\a> \a_k$ then the *-core of  $\G(F,\s)$ has 
\begin{enumerate}
\item[(i)] $\hf\r n+o(n)$ vertices in each of $\La^+,\La^-$;  
\item[(ii)] $\hf\l n+o(n)$ hyperedges with essential vertices in each of $\La^+,\La^-$;
\item[(iii)] $\hf\l e^{-\l} n +o(n)$ vertices in each of $\La^+,\La^-$ that are essential in exactly one hyperedge. 
\end{enumerate}
\end{enumerate}
\end{lemma}

Recalling that $\a_k=\inf_{x > 0}\frac{x}{(1-e^{-x})^{k-1}}$, we have that for $\a<\a_k, \r=0.$
For $\a>\a(k)$ we define $x_k(\a)$ to be the maximum $x>0$ such that 
$\a=\frac{x}{(1-e^{-x})^{k-1}}$.
It is straightforward to show that $x_k(\a)<1$, $\l=x_k(\a)$ and $\r=1-e^{-\l}$.
Thus Lemma~\ref{lcore} implies Lemma~\ref{lHp} and Lemma~\ref{lbranch}(a).
Also,  Lemmas~\ref{lconc},~\ref{lcore} imply Lemma~\ref{ldepth} as follows:

{\em Proof of Lemma~\ref{ldepth}:}  For any $\e>0$ we can choose $I$ such that $\r_I<\r+\e$.
The number of vertices outside of the *-core with *-depth greater than $I$ is $X_I^++X_I^-$ minus the size of the *-core, and hence is less than $\e n$.  This proves the lemma with $L=I$.
\proofend

We will require the following bound:

\begin{lemma}\label{lbr} For any $\a>\a_k$ there exists $\g>0$ such that
$\l e^{-\l}<(1-\g)\r/(k-1)$.
\end{lemma}

\begin{proof}  Let $x_1$ be the value of $x>0$ that minimizes $\frac{x}{(1-e^{-x})^{k-1}}$. It is straightforward to check that for $\a>\a_k$ we have $x_k(\a)>x_1$.
Differentiating, we see that 
\[(1-e^{-x_1})^{k-1}-(k-1)x_1e^{-x_1}(1-e^{-x_1})^{k-2}=0;
\qquad\mbox{ so } 1-e^{-x_1}=(k-1)x_1e^{-x_1}.\]
Clearly $\frac{1-e^{-x}}{xe^{-x}}=\inv{x}(e^x-1)=(1+\frac{x}{2}+...)$ is increasing with $x$.
So for $x>x_1$ we have $\frac{1-e^{-x}}{xe^{-x}}>k-1$ which yields the lemma since $\l=x_k(\a),\r=1-e^{-x_k(\a)}$.
\end{proof}

Note that Lemmas~\ref{lcore},~\ref{lbr} imply Lemma~\ref{lbranch}(b).

{\em Proof of Lemma~\ref{lcore}:}  We will choose a small constant  $\z>0$.
Lemma~\ref{lconc} implies that we can choose $I$ sufficiently large that, with probability at least
$1-e^{-g(n)}$:
\[(\hf\r-\z)n< X^+_I,X^-_I<(\hf\r+\z)n; \qquad Y^+_I,Y^-_I<(\hf\l+\z)n;\qquad
A^+_I,A^-_I<\hf\z n;\qquad B^+_I,B^-_I<(\hf\l e^{-\l}+\z)n.\]

Recall that the order in which vertices are removed during the *-core process does not affect the outcome.  So we can remove them as follows:  First, we carry out $I$ parallel rounds.  Then we remove vertices that are not essential in any edges one-at-a-time in arbitrary order; eg. we can pick one of the removable vertices uniformly at random, or we can choose the removable vertex with the lowest label.

After the $I$ parallel rounds, we expose the  vertices that remain, $W$, the number of hyperedges that remain, $Y_I$ and for each remaining hyperedge $f$ we expose its essential vertex, $\ess(f)$.  The following observation allows us to analyze 
$H(I)$ using the Essential Model (section~\ref{shm}).

{\bf Observation:} Consider any two hypergraphs $\Omega,\Omega'$ on the
same subset of the vertices of $H$, with edge set $\{e_1,...,e_{\ell}\}$ and $\{e'_1,...,e'_{\ell}\}$,
such that: for each $1\leq j\leq\ell$, the hyperedges $e_j,e'_j$ have the same type and the same essential vertex.  Then $\Omega,\Omega'$ are equally likely to be $H(I)$.

To see this, let $R$ be any
 hypergraph such that applying $I$ iterations of the parallel process to $R$ yields $\Omega$.  Form $R'$ from $R$ by replacing every edge of $R$ that is in $\Omega$ by the
corresponding edge from $\Omega'$. Then applying $I$ iterations of the parallel process to $R'$ will yield $\Omega'$.  Furthermore, $\pr(\G(F,\s)=R)=\pr(\G(F,\s))=R'$.

Note that this observation allows us to model $H(I)$ using the Essential Model.  So we expose the vertices of $H(I)$, and for each hyperedge $e\in H(I)$ we expose the essential vertex of $e$. 
We let $H_1$ denote the set of vertices that are essential in exactly one hyperedge; so $|H_1|=B_I^++B_I^-<(\l e^{-\l}+2\z)n$.

From here, the analysis is nearly identical to that from the proof of Lemma~\ref{lem:closure}.

Our first step will be to expose the type of every hyperedge; recall that we choose these types
independently and the probability that a hyperedge with essential vertex in $\La^s$ has type
$\t$ is $w^s(\t)$.  For each vertex $x\in H_1$,
if the type of the hyperedge in which $x$ is essential $(s,a,b)$ then we say $a(x)=a,b(x)=b$.  
We set $A=\sum_{x\in H_1}a(x)$ and $B=\sum_{x\in H_1}b(x)$. As in the proof of  Lemma~\ref{lem:closure},  with probability at least $1-e^{-g(n)}$ we have $A,B=|H_1|(\hf(k-1)+o(1))$.

As we remove vertices one-at-a-time from $H(I)$, we let $L$ denote the set of removable vertices that remain.  So initially, $|L|=A^+_I+A^-_I<\z n$.  At each step, we remove a vertex $w$ from $L$.  For each hyperedge $f$ containing $w$, if the essential vertex $\ess(f)$ is in $H_1$ then we add $\ess(f)$ to $L$. 

We carry out up to $\frac{4\z}{\g}n$ steps.  If we do not reach the *-core before that time, 
then we must have added  a total of at least $\frac{4\z}{\g}n-\z n$ vertices to $L$ during those steps.

To determine which  vertices are added to $L$ we expose the following information:
For each remaining hyperedge $f$, we ask whether $w$ is a non-essential vertex of $f$. If it is,
then we delete $f$ and place $\ess(f)$ into $L$ if $\ess(f)\in H_1$.  If
$w$ is not in $f$, then we do not expose  the non-essential vertices of $f$.

Suppose $w\in\La^+$.
As in the proof of Lemma~\ref{lem:closure}, it is easy to compute that the vertices of $H_1$
that will be added to $L$ are determined by  at most $H_1$ independent trials of total probability at most
\[\frac{A}{X_I^+-\frac{8\z}{\g}n}<\frac{(\hf\l e^{-\l}+\z)n}{(\hf\r-\z)n-\frac{8\z}{\g}n}
<1-\hf\g,\]
for $\z$ sufficiently small, by Lemma~\ref{lbr}.
 Similarly, if $w\in\La^-$ then we have at most $H_1$ independent trials of total probability at most $1-\hf\g$.

Summing over the first $\frac{4\z}{\g}n$ iterations, the total number of vertices added to $L$ is upperbounded in distribution by the sum of $\frac{4\z}{\g}n\times H_1$ independent trials, each with probability $\Theta(n^{-1})$ and with total expectation $\frac{4\z}{\g}n(1-\hf\g)=\frac{4\z}{\g}n-2\z n$.
Standard concentration results for binomial variables yield that the probability
that they total more than $\frac{4\z}{\g}n-\z n$ is at most $e^{-\d n}$
for some $\d>0$.

So for every $\z>0$, there exists $\d>0$ such that with probability at least $e^{-\d n}$, we halt within $\frac{4\z}{\g}n$ steps.  If we halt within that many steps then the number of
vertices of $\La^+$ that are in the *-core is between $X^+_I$ and $X^+_I-\frac{4\z}{\g}n$
and hence is within $\frac{5\z}{\g}n$ of $\hf\r n$.  Since we can take $\z$ arbitrarily small,
this implies that for any $g(n)=o(n)$, the number of such vertices is $\hf\r n+o(n)$ with probability at least $1-e^{-g(n)}$.  The same argument applies to the other parameters,
thus proving Lemma~\ref{lcore}.
\proofend

It only remains to prove our concentration lemma:

{\em Proof of Lemma~\ref{lconc}:}  We will apply Azuma's Inequality\cite{az} which implies (see eg. \cite{asp}) that for any random variable $Q=Q(H)=O(n)$, if changing the vertices of one of the $\a n$ hyperedges in $H$ can change $Q$ by at most an additive constant, then $\pr(|Q-\ex(Q)|>\e n)<e^{-\Theta(n)}$.

We start with the concentration of $X^+_i$.  Note that whether $v$ is counted in $X_i^+$ is
determined entirely by the subgraph induced by $N^i(v)$, the set of vertices within distance $i$ of $v$.  In an extreme case, changing a single hyperedge $f$ can affect $X_i^+$ by a lot,
if $f$ is within distance $i$ of many vertices.  So we fix a large constant $D$ and define:
\begin{eqnarray*}
\Psi_D&= &\mbox{the set of vertices that are within distance $i$ of a vertex $u$ of degree } > D,\\
X_i^+(D)&=&\mbox{the number of vertices of $\La^+\setminus \Psi_D$ that are in $H(i)$}.
\end{eqnarray*}
Changing a single hyperedge can affect $X_i^+(D)$ by at most $2k((k-1)D)^i=O(1)$. Indeed, if changing
the vertices of $f$ affects whether $v\in X^+_i$ then $v$ is connected to one of the old or new vertices of $f$ by a path of length at most $i$. If any vertex on a hyperedge of that path has degree greater than $D$ then $v\in\Psi_D$ and so $v$ will not count towards $X_i^+(D)$.  So each of the $2k$ old or new vertices of $f$ can affect at most $((k-1)D^i)$ vertices $v$.  Therefore, there exists $\eta_1=\eta_1(D,\e,k,i,\U)$ such that
\[\pr(|X_i^+(D)-\Exp(X_i^+(D))|>\inv{3}\e n)<e^{-\eta_1n}.\]
A standard property of random graphs (and indeed an easy calculation) yields that by taking $D$ sufficently large, we can make $\ex(|\Psi_D|)$ an arbitrarily small multiple of $n$.  (Roughly: the expected number of vertices $u$ of degree greater than $D$ drops exponentially in $D$ while the expected number of
vertices within distance $i$ of each such $u$ is linear in $D$, for fixed $i$.)  So we choose $D$  
such that $\ex(|\Psi_D|)<\inv{3}\e n$.

Next we show that $|\Psi_D|$ is concentrated.  A similar argument to that above shows that
changing the vertices of a single hyperedge $f$ can affect $|\Psi_D|$ by at most $2k((k-1)D)^i=O(1)$.  Indeed, if changing $f$ affects whether $v\in\Psi_D$ then  $v$ is connected to one of the old or new vertices of $f$ by a path of length at most $i$.  If any vertex on the hyperedges of that path has degree greater than $D$ then $v\in\Psi_D$ regardless of
the choice of $f$.  So each of the $2k$ old or new vertices of $f$ can affect at most $((k-1)D^i)$ vertices $v$.
Therefore, there exists $\eta_2=\eta_2(D,\e,k,i,\U)$ such that
\[\pr(|\Psi_D-\Exp(|\Psi_D|)|>\inv{6}\e n)<e^{-\eta_2n}.\]
Noting that $X_i^+(D)<X_i^+<X_i^+(D)+|\Psi_D|$ and applying linearity of expectation,
we have
\[\pr(|X_i^+-\Exp(X_i^+)|>\e n)<e^{-\eta_1n}+e^{-\eta_2n}<e^{-\eta n},\]
for any $\eta<\eta_1,\eta_2$.  The proof for the remaining parameters is nearly identical.
\proofend

We close this section by noting that by the same reasoning as for the Observation in the proof of Lemma~\ref{lcore},
we can model the *-core of $\G(F,\s)$ using the Essential Model.  We do so in section~\ref{sec:stable}.

\section{Frozen variables in the *-core} \label{sec:stable}
Here, we prove Lemma~\ref{lstarfr}(a).  Recall that $(F,\s)$ is drawn from $P(\U,n,M=rn)$
where $\U$ is symmetric and 1-essential.

$H^*$ is the *-core of $\G(F,\s)$, so every edge of $H^*$ has exactly one essential vertex
and every vertex is essential for at least one edge.   $H_1$ is the set of vertices
that are essential in exactly one hyperedge of $H^*$.
We need to show that, with sufficiently high probability, all but $o(n)$  vertices in $H^*$ are
$\b n$-frozen variables of $(F,\s)$.  

\begin{definition}
For each vertex $x\in H_1$, we use $e(x)$ to denote the unique hyperedge of $H^*$ in which $x$ is the essential vertex.
\end{definition}

\begin{definition}
 A \emph{flippable set} of $H^*$ is a set of vertices $S\subset H^*$ such that  for every $x\in S$ and for every hyperedge $f\in H^*$ in which $x$ is essential, $S$ contains another vertex of $f$.
\end{definition}

Note that, since every hyperedge of $H^*$ has exactly one essential variable, that other vertex is not essential for $f$.

Given two boolean assignments $\s,\s'$ to the variables of $F$, we let $\s\Delta\s'$ denote the
set of variables $x$ for which $\s(x)\neq\s'(x)$.

\begin{proposition}\label{prop:flippable}
\begin{enumerate}
\item[(a)]\label{prop:flippable1} If $\s'$ is any  solution of $F$, then $(\s\Delta\s')\cap H^*$ is a flippable set.
\item[(b)]\label{prop:flippable2} The union of any collection of flippable sets is a flippable set.
\end{enumerate}
\end{proposition}

\begin{proof}
For (a):  if $(\s\Delta\s')\cap H^*$ is not a flippable set, then there is some $x\in(\s\Delta\s')\cap H^*$ and a hyperedge $f\in H^*$ such that $x$ is essential for $f$
and $\s\Delta\s'$ contains no other vertices of $f$.  Since $H^*\subset\G(F,\s)$, this means
that $x$ is essential for the constraint corresponding to $f$ in $(F,\s)$, and that $\s'$ changes
the value of $x$ but not of any other variables in $f$.  Therefore $\s'$ violates $f$ and so $\s'$ is not a solution for $F$.

For (b): this is immediate from the definition of a flippable set.
\end{proof}

To prove Lemma ~\ref{lstarfr}(a), we will show that there exists $\phi'(n)=o(n)$, $\zeta>0$ such that,
with sufficiently high probability,  there are no flippable sets $S$ in $H^*$ of size $\phi'(n)\leq|S|\leq \zeta n$. We will apply a first moment bound.  A direct approach does not work,
because of a ``jackpot phenomena":  The existence of a flippable set $S$ typically implies the existence
of an exponential number of other flippable sets formed by adding to $S$ variables  $x\in H_1$ with the property that $e(x)$ contains a member of $S$.  To overcome this issue, we focus instead on sets
with the following property.

\begin{definition}
We say that a set $A\subseteq H^*\setminus H_1$ is \emph{weakly flippable} if there exists $P\subseteq H_1$ such that $A\cup P$ is flippable.  $A$ is \emph{$\psi$-weakly flippable} if there exists such a $P$ with $|P|\leq\psi$.
\end{definition}

Given a flippable set $S\subseteq H^*$, we consider a directed graph $D(S)\subseteq D$. 
The vertices of $D(S)$ are the vertices of $S$; the edges of $D(S)$ are defined as follows:
\begin{itemize}
\item For each $x\in S\cap H_1$, we choose one other variable $x'\in e(x)$ that is in $S$,
and we add the  edge $x\edgw x'$ to  $D(S)$.
\end{itemize}
Note that, since $S$ is flippable, there is at least one such $x'$.  It is not important which one we choose,
but to be specific we could, eg., choose the lowest indexed variable from amongst all variables of $S$ (other than $x$) in $e(x)$.

Thus, every vertex in $D(S)$ has outdegree either 0 or 1. We define:

\begin{itemize}
\item $A_S=S\setminus H_1$. Note that $A_S$ is the set of vertices with outdegree $0$ in $D(S)$. 
\item  $C_S$ is the set of all vertices on directed cycles of $D(S)$. Note that those directed cycles are disjoint
since the maximum outdegree is 1.
\end{itemize}

Since the outdegree of every vertex outside of $A_S$ is one,  there is a directed path from every $x\in S\setminus(A_S\cup C_S)$ to $A_S\cup C_S$.

\begin{definition}
A set of vertices $x_{1},\ldots,x_{l}\in H_1$ is \emph{cyclic} if for some permutation $\pi \in \mathcal{S}_l$, $x_{{\pi(j)}}$ is in $e(x_j)$ for every $1\leq i\leq a$.
\end{definition}
\begin{definition}\label{dcl}
Given a set $A\subseteq H^*$, the \emph{closure} of $A$, $\cl{A}$ is the set of all vertices $x$ such that, either 
\begin{enumerate}
\item[(a)] $x\in A$, or 
\item[(b)]  $x \in H_1\setminus A$ and there is a sequence $x=x_0,x_1,...,x_{\ell}$ where (i) $x_{\ell}\in A$ and (ii) for all $i<\ell$: $x_i\in H_1\setminus A$ and
$x_{i+1}\in e(x_i)$.
\end{enumerate}
\end{definition}

\begin{proposition}\label{prop:flippableweak}
If $S$ is a flippable set, then:
\begin{enumerate}
\item[(a)] $A_S$ is weakly flippable. 
\item[(b)] $C_S$ is cyclic.
\item[(c)] $S \subseteq \cl{A_S \cup C_S}$.
\end{enumerate}
\end{proposition}

\begin{proof}
(a) follows from the definition of weakly flippable, with $P=S\cap H_1$.  

(b) follows from the definition of cyclic, where the directed cycles of $D(S)$ form $\pi$.

For (c), if $x\in S\setminus (A_S\cup C_S)$, then $x\in H_1$ and the directed path from $x$
to $A_S\cup C_S$ in $D(S)$ indicates that $x$ satisfies condition (b) of Definition~\ref{dcl}.
\end{proof}

\begin{lemma}\label{labc}
Suppose that for some $\phi,\phi',\psi$ we have:
\begin{enumerate}
\item[(a)]\label{lem:cond1} There is no $\psi$-weakly flippable set $A\subseteq H^*\setminus H_1$ such that $ \phi<|A|<\psi$.
\item[(b)]\label{lem:cond2} There is no cyclic set $C$ such that $\phi<|C|<\psi$.
\item[(c)]\label{lem:cond3} There is no set $A$ such that $|A|\leq 2\phi$ and $|\cl{A}|> \phi'$.
\end{enumerate}
Then there is no flippable set $S\subseteq H^*$  such that  $\phi'<|S|<\psi$.
\end{lemma}
\begin{proof} We apply Proposition~\ref{prop:flippableweak}. Let $S$ be a flippable set with $|S|<\psi$. Then $A_S$ is $\psi$-weakly flippable. Thus, by (a), $|A_S|\leq\phi$. Since $C_S$ is cyclic and $|C_S|\leq |S|< \psi$, (b) implies $|C_S|\leq \phi$. Therefore, $|A_S\cup C_S|\leq2\phi$, which by (c) implies that $|S|\leq|\cl{A_S\cup C_S}|\leq\phi'$. The lemma follows.
\end{proof}

The following lemmas establish that the conditions of Lemma~\ref{labc} hold with sufficiently high probability. 

\begin{lemma}\label{lem:weakly} There exists $\z=\z(\U,\a)>0$, and for any $g(n)=o(n)$, there exists $\phi(n)$ satisfying
$g(n)<<\phi(n)=o(n)$  such that:\\
The probability that there is a $(\z n)$-weakly flippable set $A$ of $H^*$ with $ \phi(n)<|A|<\z n$ is at most $e^{-g(n)}$.
\end{lemma}


\begin{lemma}\label{lem:cyclic} There exists $\z=\z(\U,\a)>0$, and for any $g(n)=o(n)$, there exists $\phi(n)$ satisfying
$g(n)<<\phi(n)=o(n)$ such that:\\
The probability that there is a cyclic set $C$ in $H^*$ with $ \phi(n)<|C|<\z n$ is at most $e^{-g(n)}$.
\end{lemma}


\begin{lemma}\label{lem:closure} There exists $\z=\z(\U,\a)>0$, and for any $\phi(n)=o(n)$, there exists 
$\phi'(n)=o(n)$ such that:
The probability that there is a  set $A\subset H^*$ with $|A|<2\phi(n) $ and $|\cl{A}|>\phi'(n)$ is at most $e^{-\phi'(n)}$.
\end{lemma}

These lemmas yield Lemma~\ref{lstarfr}(a) as follows:

{\em Proof of Lemma~\ref{lstarfr}(a):} Note that we can take $\phi'(n)>g(n)$.
Lemmas~\ref{labc},~\ref{lem:weakly},~\ref{lem:cyclic},~\ref{lem:closure} imply
that for all $g(n)=o(n)$, there exists $\phi'(n)=o(n)$ such that with probability at least $1-3e^{-g(n)}$ the *-core $H^*$ of
$\G(F,\s)$ has no flippable set of size between $\phi'(n)$ and $\z n$.    So suppose that there is no such flippable set in $H^*$.

Let $S_1,...,S_t$ be all flippable sets in $H^*$ of size less than $\z n$.  Thus each $|S_i|<\phi'(n)$.
Assume by induction that $|\cup_{i=1}^jS_i|<\phi'(n)$.  Then $|\cup_{i=1}^{j+1}S_i|<2\phi'(n)<\z n$.
By Proposition~\ref{prop:flippable}(b), $\cup_{i=1}^{j+1}S_i$ is a flippable set and hence it
must have size less than $\phi'(n)$.  Therefore $|\cup_{i=1}^{t}S_i|<\phi'(n)$.

Now consider any sequence of solutions $\s=\s_0,\s_1,...,\s_{\ell}$ in which the assignment changes for at least one variable in $H^*\setminus(\cup_{i=1}^{t}S_i)$.  Let $i$ be the lowest index so that $\s_i(x)\neq\s(x)$ for some $x\in H^*\setminus(\cup_{i=1}^{t}S_i)$.  Therefore
$x\in(\s_i\Delta\s)\cap H^*$ which, by  Proposition~\ref{prop:flippable}(a), is a flippable set.
Since $x\notin \cup_{i=1}^{t}S_i$, this implies $|(\s_i\Delta\s)\cap H^*|\geq\z n$.  By our choice of $i$, $|(\s_i\Delta\s_{i-1})\cap H^*|\geq |(\s_i\Delta\s)\cap H^*|-|\cup_{i=1}^{t}S_i|\geq\z n-\phi'(n)$. Therefore every variable in $H^*\setminus(\cup_{i=1}^{t}S_i)$ is $(\z n-\phi'(n))$-frozen. This yields Lemma~\ref{lstarfr}(a) for  any $\b<\z$ after rescaling $g(n)$.
\proofend

We prove Lemmas~\ref{lem:weakly},~\ref{lem:cyclic},~\ref{lem:closure}
in the next three subsections.  In each case, we will study $H^*$ using the Essential Model.
See the discussion at the end of Appendix~\ref{a3} explaining why it is valid to do so.

\subsection{Weakly-flippable sets: Proof of Lemma \ref{lem:weakly}}


Suppose that $A\subseteq H^*\setminus H_1$ is a $(\z n)$-weakly flippable set with $\phi(n)<|A|<\z n$.

Set $a:=|A|$ and  note that there are at least $2a$ hyperedges of $H^*$ whose essential variables are in $A$, since $A$ contains
no variables of $H_1$.  Let $e_1,...,e_{2a}$
denote exactly $2a$ such hyperedges;
to be specific, the $2a$ with the lowest indices. Since $A$ is $(\z n)$-weakly flippable, for each $1\leq j\leq 2a$ there exists a sequence of vertices $x_{j,0},x_{j,1},...,x_{j,t_j}$ such that:
\begin{enumerate}
\item[(i)] $x_{j,0}\in A$ is the essential vertex of $e_j$;
\item[(ii)] $x_{j,t_j}\in A$;
\item[(iii)]  $x_{j,1}\in e_j$, and if $t_j>1$ then for each $1\leq i\leq t_j-1$: $x_{j,i}\in H_1$
and $x_{j,i+1}\in e(x_{j,i})$.
\end{enumerate}
Note that possibly $t_j=1$ in which case $e_j$ contains a non-essential member of $A$.
 
These sequences  are not necessarily disjoint. However, since $e_j\neq e_{j'}$ for all $j\neq j'$, we can take initial portions of them so that the portions in $H_1$ are disjoint.  I.e.,  there exist $l_1,\ldots,l_{2a}\geq0$ with $\sum_{j=1}^{2a}l_j\leq\z n$ such that 
\begin{enumerate}
\item[(i)] the vertices
$x_{j,i}: 1\leq j\leq 2a, 1\leq i\leq l_j$ are distinct;
\item[(ii)] for $j=1,\ldots,2a$:
$x_{j,l_j+1}\in A\cup\{x_{j',i}:1\leq j'<j, 1\leq i\leq l_{j'}\}$.
\end{enumerate}

We will bound the expected number of such collections of sequences, when $H^*$ is chosen
from the Essential Model. So we expose the vertices of $H^*$, and for each hyperedge $e_1,...,e_{\a n}$ we expose the essential vertex of $e_i$.  By Lemma~\ref{lbranch}(a)
we can assume that $|H^*\cap\La^+|,|H^*\cap\La^+|=\hf|H^*|+o(n)$.

Fix some $0\leq\ell\leq\z n$.  First we will choose $l_1,...,l_{2a}\geq0$ summing to $\ell$.
The number of choices is $\ell+2a-1\choose 2a-1$.  

Next, we choose $A$; note that this determines $e_1,...,e_{2a}$ and their essential vertices
$x_{1,0},...,x_{2a,0}$. The number of choices is  ${|H^*|\choose a}\leq {n\choose a}$. 

Next we choose the remaining vertices.  To do so, we first determine
their signs; i.e. which are in $\La^+$ and which are in $\La^-$.  So for each $j$, we choose a pattern $\theta_j$ -  a sequence
of $l_j+2$ terms from $\{+,-\}$ indicating the signs of $x_{j,0},...,x_{j,l_j+1}$.  Note that,
since $x_{j,0}$ is already determined, the first sign of $\theta_j$ is already known.  
Recall from Lemma~\ref{lbranch} that we can assume $|H_1^+|,|H_1^-|<|H^*|\times\frac{\hf-\g}{k-1}$. Thus, given $\theta_j$, the number of choices for $x_{j,1},...,x_{j,l_j}$ is at most $(|H^*|\times\frac{\hf-\g}{k-1})^{l_j}$.

Finally, we choose $x_{j,l_j+1}:  1\leq j\leq 2a$. These are not neccesarily distinct, and they are all members of $A\cup\{x_{j,i}:1\leq j<2a, 1\leq i\leq l_{j}\}$.  So the number of choices is
at most $(a+\ell)^{2a}$.

Having selected these vertices, we bound the probability that all edges are as required.

Consider selecting the type of a hyperedge whose essential vertex is in
$\La^+$; thus we are select type $\t= (1,a,b)$ with probability $w^{+1}(\t)$.  We let $g=g(\U)$ denote the
expected value of $a$, and so the expected value of $b$ is $k-1-g$.  Because $\U$ is symmetric and $|H^*\cap\La^+|=|H^*\cap\La^-|(1+o(1))$, it follows that $w^{+1}(1,a,b)=w^{-1}(-1,b,a)+o(1)$. So when we select the type of a hyperedge whose essential vertex is in $\La^-$, the expected value of $b$ is $g+o(1)$.

Now we select the types and then the non-essential vertices for the hyperedges $e_1,...,e_{2a}$ and $e(x_{j,i}):1\leq j\leq 2a, 1\leq i\leq l_j$.  Recall that we choose those vertices uniformly from
$\La^+\cap H^*$ or $\La^-\cap H^*$ depending on what the type of the hyperedge tells us the sign of the
vertex should be.  For each $j$, we require that $x_{j,1}$ is a non-essential vertex of $e_j$ and that $x_{j,i+1}$ is a non-essential vertex of $e(x_{j,i})$.  By Lemma~\ref{lbranch}(a), $|\La^+\cap H^*|,|\La^-\cap H^*|=|H^*|(\hf+o(1))$, and so for each hyperedge this event occurs with probability $\frac{2g+o(1)}{|H^*|}$ if the essential
and non-essential vertices have the same sign, and $\frac{2(k-1-g)+o(1)}{|H^*|}$ otherwise.
Note also that these events are independent.

We let $y(\theta_j)$ denote the number of terms in $\theta_j$ that are the same as the preceding term. So the probability that the required vertices are selected as non-essential vertices in each hyperedge is:
\[\left(\frac{2g+o(1)}{|H^*|}\right)^{\sum_{j=1}^{2a}y(\theta_j)}
\left(\frac{2(k-1-g)+o(1)}{|H^*|}\right)^{\sum_{j=1}^{2a}l_j+1-y(\theta_j)}.\]

Putting this all together yields that the expected number of $(\z n)$-weakly flippable sets $A$,
given $a,\ell$, is at most:
\begin{eqnarray}
\nonumber&&{\ell+2a-1\choose 2a-1}{n\choose a}
\left(|H^*|\times\frac{\hf-\g}{k-1}\right)^{\sum_jl_j}(a+\ell)^{2a}\\
\nonumber&&\qquad\times
\sum_{\theta_1,...,\theta_{2a}}
\left(\frac{2g+o(1)}{|H^*|}\right)^{\sum_{j=1}^{2a}y(\theta_j)}
\left(\frac{2(k-1-g)+o(1)}{|H^*|}\right)^{\sum_{j=1}^{2a}l_j+1-y(\theta_j)}\\
\nonumber&<&{\ell+2a\choose 2a}{n\choose a}
\left(|H^*|\times\frac{\hf-\g}{k-1}\right)^{\ell}(a+\ell)^{2a}
\left(\frac{2+o(1)}{|H^*|}\right)^{2a+\ell}
\sum_{\theta_1,...,\theta_{2a}}
g^{\sum_{j=1}^{2a}y(\theta_j)}
(k-1-g)^{\sum_{j=1}^{2a}l_j+1-y(\theta_j)}\\
&<&\left(\frac{e(\ell+2a)}{2a}\right)^{2a}\left(\frac{en}{a}\right)^a
\left(\frac{3(a+\ell)}{|H^*|}\right)^{2a}\left(\frac{1-\g}{k-1}\right)^{\ell}
\prod_{j=1}^2a\sum_{\theta_j}
g^{y(\theta_j)}
(k-1-g)^{l_j+1-y(\theta_j)}.\label{esumprod}
\end{eqnarray}

For each value of $y$, there are ${l_j+1\choose y}$ patterns $\theta_j$ with $y(\theta_j)=y$,
since the first sign in $\theta_j$ is already chosen.
This implies 
\[\prod_{j=1}^2a\sum_{\theta_j}
g^{y(\theta_j)}
(k-1-g)^{l_j+1-y(\theta_j)}
=\prod_{j=1}^{2a}\sum_{y=0}^{l_j+1}{l_j+1\choose y}g^y(k-1-g)^{l_j+1-y}
=(k-1)^{\ell}.\]
So (\ref{esumprod}) is at most
\[\left(\frac{e(\ell+2a)}{2a}\right)^{2a}\left(\frac{en}{a}\right)^a
\left(\frac{3(a+\ell)}{|H^*|}\right)^{2a}\left(\frac{1-\g}{k-1}\right)^{\ell}(k-1)^{\ell}
<\left(1+\frac{\ell}{2a}\right)^{2a}\left(1+\frac{\ell}{a}\right)^{2a}
\left(\frac{Ca}{n}\right)^{a}(1-\g)^{\ell},\]
for some constant $C>9e^3(\frac{n}{|H^*|})^2$ (see Lemma~\ref{lHp}).  So the total expected number of  $A$ with $\phi(n)<|A|<\z n$ is at most:
\begin{eqnarray*}
\sum_{a=\phi(n)}^{\z n}\left(\frac{Ca}{n}\right)^{a}
\sum_{\ell\geq 0}\left(1+\frac{\ell}{a}\right)^{4a}(1-\g)^{\ell}.
\end{eqnarray*}
To bound this, it is easy to see that $\left(1+\frac{\ell}{a}\right)^{4a}(1-\frac{\g}{2})^{\ell}$
is maximized at $\ell=O(a)$ and hence is at most $Y^{4a}(1-\frac{\g}{2})^{4a}<Y^{4a}$ for some constant $Y=Y(\g)$.  Since  $1-\g<(1-\frac{\g}{2})^2$, this yields an upper bound of:
\[\sum_{a=\phi(n)}^{\z n}\left(\frac{Ca}{n}\right)^{a}
\sum_{\ell\geq 0}Y^{4a}(1-\frac{\g}{2})^{\ell}=O(1)\left(\frac{CY^4a}{n}\right)^{a}.\]
By taking $\z<\inv{2CY^4}$, this is less than $\sum_{a=\phi(n)}^{\z n}2^{-a}$ which is
 less than $e^{-g(n)}$ for any $\phi(n)>>g(n)$.
\proofend

\subsection{Cyclic sets: Proof of lemma \ref{lem:cyclic}}\label{scs}


Note that the vertices of a cyclic set are partitioned into cycles by the permutation $\pi$.
We will fix a constant $Z$, to be named later.  A {\em small-cyclic} set is a cyclic set in which each cycle has length at most $Z$. A {\em large-cyclic} set is a cyclic set in which each cycle has length greater than $Z$.

\begin{lemma} For any $g(n)=o(n)$ there exists $\phi(n)$ such that with probability at least
$1-e^{-g(n)}$:
\begin{enumerate}
\item[(a)] $H^*$ has no small-cyclic sets of size at least $\hf\phi(n)$.
\item[(b)] $H^*$ has no large-cyclic sets of size at least $\hf\phi(n)$.
\end{enumerate}
\end{lemma}

This clearly proves Lemma~\ref{lem:cyclic} as any cyclic set of size at least $\phi(n)$
contains either a small-cyclic set or a large-cyclic set of size at least $\hf\phi(n)$.  Again, we work in the Essential Model.

{\em Proof of (a):} We say that a {\em cycle} in $\G(F,\s)$ is a set of vertices $x_1,...,x_{\ell}$
such that $x_i,x_{i+1}$ lie in a common hyperedge of $\G(F,\s)$ for each $i$ (addition is mod $\ell$). For any $x_i,x_j$,
the probability that $x_i,x_j$ share a hyperedge in $\G(F,\s)$ is less than $c/n$,
for some constant $c=c(\U,\a)$.

 If $H^*$ has a small-cyclic set of size at least $\hf\phi(n)$,
then the hypergraph $\G(F,c)$ must contain at least $\phi(n)/(2Z)$ cycles of size
at most $Z$, and so it must contain at least $\phi(n)/(2Z^2)$ cycles of size
exactly $z$ for some $z\leq Z$.   Setting $W:=\phi(n)/(2Z^2)$, the probability of this occurring for $z$ is less than:
\[\frac{n^{zW}}{W!}\left(\frac{c}{n}\right)^{zW}=\frac{(c^{z})^{W}}{W!}<\inv{2Z}e^{-g(n)},\]
if $\phi(n)>>g(n)$.  (Note that the dependency between the events that the $zW$ pairs of
vertices each share a hyperedge goes in the right direction for this bound to hold.)
Summing over all $z\leq Z$ proves (a).
\proofend

{\em Proof of (b):}  We will bound the expected number of large cyclic sets of size $a$.

A pattern $\theta$ is a sequence
of $a$ terms from $\{+,-\}$ indicating the signs of $x_1,...,x_a$.   By Lemma~\ref{lbranch}, we can assume that $|H_1^+|,|H_1^-|<\frac{\hf-\g}{k-1}|H^*|$.  So for any pattern $\theta$, the number of choices
for $x_1,...,x_a$ is at most $\left(\frac{\hf-\g}{k-1}|H^*|\right)^a$.  

Given a pattern $\theta$
and a permutation $\pi$, we let $y(\theta,\pi)$ denote the  number of $i$ such that $x_i,x_{\pi(i)}$ have the same sign.  Recall $g$ from the proof of Lemma~\ref{lem:weakly}; the same reasoning as in that proof says that, for any choice of $x_1,...,x_a$ in agreement with $\theta$, the probability that $x_{\pi(i)}$ is a non-essential vertex in $e(x_i)$ for every $i$ is
$\left(\frac{2g+o(1)}{|H^*|}\right)^{y(\theta,\pi)} \left(\frac{2(k-1-g)+o(1)}{|H^*|}\right)^{a-y(\theta,\pi)}$.

We let $c(\pi)$ denote the number of cycles in $\pi$; since we are considering large-cyclic sets, we only need to consider permutations $\pi$ with $c(\pi)<a/Z$. For any $\pi,y$, the number of choices of $\theta$ with $y(\theta,\pi)=y$ is at most $2^{c(\pi)}{a\choose y}<2^{a/Z}{a\choose y}$.  Indeed, there are ${a\choose y}$ choices of the values of $i$ for which $x_i,x_{\pi(i)}$ have the same sign; given one such choice, the pattern is determined by fixing the sign of one vertex
in each of the $c(\pi)$ cycles.  Note that this is an upper bound; as for some $\pi,y$, parity
conditions will imply that there is no such $\theta$.

To bound the expected number of large-cyclic sets of size $a$, we sum over all ordered choices
of $x_1,...,x_a$ and all choices of $\pi$, and then divide by $a!$, obtaining:
\begin{eqnarray*}
&&\inv{a!}\sum_{\theta}\left(\frac{\hf-\g}{k-1}|H^*|\right)^a
\sum_{\pi}\left(\frac{2g+o(1)}{|H^*|}\right)^{y(\theta,\pi)} 
\left(\frac{2(k-1-g)+o(1)}{|H^*|}\right)^{a-y(\theta,\pi)}\\
&<&\left(\frac{1-\g}{k-1}\right)^a
\inv{a!}\sum_{\pi}\sum_{y=0}^a2^{a/Z}{a\choose y}g^y(k-1-g)^{a-y}\\
&\leq&\left(\frac{1-\g}{k-1}2^{1/Z}\right)^a(k-1)^a\\
&<&(1-\hf\g)^a
\end{eqnarray*}
if $Z$ is chosen large enough that $(1-\g)2^{1/Z}<(1-\hf\g)$.

So the probability that there is a large-cyclic set of size at least $\hf\phi(n)$ is at most
$O(1)(1-\hf\g)^{\hf\phi(n)}<\hf e^{-g(n)}$ for any $\phi(n)>>g(n)$.
\proofend

\subsection{Closure: Proof of lemma \ref{lem:closure}}  We will choose $\phi'(n)>>\phi(n)$.
Again, we work in the Essential Model.

Consider  a set $A$ of size at most $2\phi(n)=o(\phi'(n))$.  We can find $\cl{A}$ using the following search:

\begin{enumerate}
\item Initialize $C=\emptyset, L=A$.
\item While $L\neq\emptyset$
\begin{enumerate}
\item Choose $u\in L$.
\item For every $w\in H_1\setminus(C\cup L)$ such that $u\in e(w)$, add $w$ to $L$.
\item Remove $u$ from $L$ and add $u$ to $C$.
\end{enumerate}
\end{enumerate}

When this procedure halts, $C=\cl{A}$. Note that $|\cl{A}|$ is the number of times that
we execute the loop in Step 2.   If $|\cl{A}|>\phi'(n)$ then during the first $\phi'(n)$ iterations
we never reach $L=\emptyset$ and so we must add a total of more than $\phi'(n)-|A|=\phi'(n)(1-o(1))>\phi'(n)(1-\hf\g)$ vertices to $L$ in Step 2(b), where $\g$ comes from Lemma~\ref{lbranch}.
We will  bound the probability of that occuring. 

We analyze $H^*$ using the Essential Model. So we expose the vertices of $H^*$, and for each hyperedge $e_1,...,e_{\a n}$ we expose the essential vertex of $e_i$.  By Lemma~\ref{lbranch}(a)
we can assume that $|H^*\cap\La^+|,|H^*\cap\La^+|=\hf|H^*|+o(n)$.

Our first step will be to expose the type of every hyperedge; recall that we choose these types
independently and the probability that a hyperedge with essential vertex in $\La^s$ has type
$\t$ is $w^s(\t)$.  For each vertex $x\in H_1$,
if the type of $x$ is chosen to be $(s,a,b)$ then we say $a(x)=a,b(x)=b$.  
We set $A=\sum_{x\in H_1}a(x)$ and $B=\sum_{x\in H_1}b(x)$. 

Because $\U$ is symmetric and $|H^*\cap\La^+|=|H^*\cap\La^-|(1+o(1))$, it follows that $w^{+1}(1,a,b)=w^{-1}(-1,b,a)+o(1)$.  This implies that for $x\in\La^+,y\in\La^-$, $\ex(a(x))=\ex(b(y))+o(1)$ and $\ex(b(x))=\ex(a(y))+o(1)$, and it follows that $\ex(A),\ex(B)=|H_1|(\hf(k-1)+o(1))$.
The number of hyperedges of each type is a binomial variable and so is easily seen to be highly enough concentrated that with probability at least $1-e^{-g(n)}$ we have $A,B=|H_1|(\hf(k-1)+o(1))$.

Now we analyze our search.  We can choose $u$ arbitrarily in Step 2(a); to be specific, we choose the $u\in L$ with the lowest index.  To carry out Step 2(b): for every $x\in H_1\setminus(C\cup L)$, we expose whether $u\in e(x)$; if $u\notin e(x)$ then we do not expose the non-essential vertices of $e(x)$.  

Suppose $u\in\La^+$. To test whether $u\in e(x)$, we ask whether $u$ is one of the $a(x)$
non-essential variables from $\La^+$.  Initially, the probability is $\frac{a(x)}{|H^*\cap\La^+|}$;
as the procedure progresses, this increases as we  have exposed that the members of $C$ are not in $e(x)$. But since $|C|\leq\phi'(n)$ it never exceeds 
$\frac{a(x)}{|H^*\cap\La^+|-\phi(n)}$. We ask this for every $x\in H_1\setminus(C\cup L)$
resulting in at most $H_1$ independent trials of total probability at most
\[\frac{A}{|H^*\cap\La^+|-\phi(n)}=\frac{|H_1|(\hf(k-1)+o(1))}{\hf |H^*|(1+o(1))}
<1-\g,\]
by Lemma~\ref{lbranch}(b).
 Similarly, if $u\in\La^-$ then we have at most $H_1$ independent trials of total probability at most $1-\g$.

Summing over the first $\phi'(n)$ iterations, the total number of vertices added to $L$ is upperbounded in distribution by the sum of $\phi'(n)H_1$ independent trials, each with probability $\Theta(n^{-1})$ and with total expectation $\phi'(n)(1-\g)$.
Standard concentration results for binomial variables yield that the probability
that they total more than $\phi'(n)(1-\hf\g)$ is at most $e^{-c\phi'(n)}$ for some $c=c(g,k,\g)$.

So the expected number of sets $A$ of size at most $\phi(n)$ for which $|\cl{A}|\geq\phi'(n)$
is at most
\[\sum_{a=1}^{\phi(n)}{|H^*|\choose a} e^{-c\phi'(n)}<\phi(n){n\choose\phi(n)}e^{-c\phi'(n)}
<\phi(n)\left(\frac{ene^{-c\phi'(n)/\phi(n)}}{\phi(n)}\right)^{\phi(n)}.\]
By choosing $\phi(n)\log(n/\phi(n))<<\phi'(n)=o(n)$, this probability is less than $e^{-\phi(n)}$,
as required.
\proofend

\end{document}